\documentclass{sn-jnl}
\usepackage[utf8]{inputenc}
\usepackage{adjustbox}
\usepackage{graphicx}
\usepackage{multirow}
\usepackage{amsmath,amssymb,amsfonts}
\usepackage{amsthm}
\usepackage{mathtools}
\usepackage{mathrsfs}
\usepackage{natbib}
\usepackage[title]{appendix}
\usepackage{xcolor}
\usepackage{textcomp}
\usepackage{manyfoot}
\usepackage{booktabs}
\usepackage{algorithm}
\usepackage{algpseudocode}
\usepackage{listings}
\usepackage{comment}
\usepackage{breakurl}
\numberwithin{equation}{section}
\usepackage{longtable}
\usepackage{geometry} 
\usepackage[labelfont=bf,labelsep=space]{caption}
\usepackage{subfigure}
\usepackage{cases}
\usepackage{ifthen}
\usepackage{rotating}
\usepackage{pdflscape}
\theoremstyle{definition} 

\usepackage{lipsum} 
\usepackage{tabularx}
\newtheorem{theorem}{Theorem} 

\newtheorem{remark}{Remark}

\begin{document}

\title[Testing Against Tree Ordered Alternatives in One-way ANOVA]{Testing Against Tree Ordered Alternatives in One-way ANOVA}


\author[1]{\fnm{Subha} \sur{Halder}}\email{sb.halder12345@gmail.com}

\author[2]{\fnm{Anjana} \sur{Mondal}}\email{anjanam@iittp.ac.in}
\equalcont{These authors contributed equally to this work.}

\author*[1]{\fnm{Somesh} \sur{Kumar}}\email{smsh@maths.iitkgp.ac.in}
\equalcont{These authors contributed equally to this work.}

\affil[1]{\orgdiv{Department of Mathematics}, \orgname{Indian Institute of Technology Kharagpur}, \orgaddress{\city{Kharagpur}, \postcode{721302}, \state{West Bengal}, \country{India}}}

\affil[2]{\orgdiv{Department of Mathematics and Statistics}, \orgname{Indian Institute of Technology Tirupati}, \orgaddress{\city{Tirupati}, \postcode{517619}, \state{Andhra Pradesh}, \country{India}}}

\affil[3*]{\orgdiv{Department of Mathematics}, \orgname{Indian Institute of Technology Kharagpur}, \orgaddress{\city{Kharagpur}, \postcode{721302}, \state{West Bengal}, \country{India}}}

\abstract{
\{The likelihood ratio test against a tree ordered alternative in one-way heteroscedastic ANOVA is considered for the first time. Bootstrap is used to implement this and two multiple comparisons based tests and shown to have very good size and power performance.\} \\
 In this paper, the problem of testing the homogeneity of mean effects against the tree ordered alternative is considered in the heteroscedastic one-way ANOVA model. The likelihood ratio test and two multiple comparison-based tests - named Max-D and Min-D are proposed and implemented using the parametric bootstrap method. An extensive simulation study shows that these tests effectively control type-I error rates for various choices of sample sizes and error variances. Further, the likelihood ratio and Max-D tests achieve very good powers in all cases. The test Min-D is seen to perform better than the other two for some specific configurations of parameters.  The robustness of these tests is investigated by implementing some non-normal distributions, viz., skew-normal, Laplace, exponential, mixture-normal, and t distributions. `R' packages are developed and shared on ``Github" for the ease of users.  The proposed tests are illustrated on a dataset of patients undergoing psychological treatments.}



\maketitle

\section{Introduction}\label{Introduction}
The problem of determining whether the average effects of new treatments are better than that of the control group occurs frequently in various experiments related to clinical trials, dose-response studies, design of industrial equipment, etc.  
In standard clinical trials, several new treatment regimes are administered to different groups of patients, and the efficacy of these regimes is to be compared with a control group that received a placebo. Formally, let $\mu_1, \mu_2,\ldots,\mu_k$ be the mean effects of the $k$ treatments, and let $\mu_0$ be the mean effect of the control group. Then, one is interested in testing the hypothesis

\begin{equation}
\begin{aligned}
     H_0 &: \mu_0 = \mu_1 = \cdots = \mu_k ~\text{against}\\
      H_1 &: \mu_0 \leq \mu_i \;\; \forall i = 1,2,\ldots, k \; (k \geq 3), \text{ with at least one strict inequality.}\label{eq:1.1}
\end{aligned}
\end{equation}

\noindent The data for this problem can be represented using a one-way ANOVA fixed effects model  
\begin{equation}
    X_{ij} = \mu_i + \epsilon_{ij} , j = 1,2,\ldots,n_i; i = 0,1,\ldots,k.
\end{equation}
\noindent Here $X_{ij}$ represents the response corresponding to the $j$-th replication of the $i$-th treatment, $\mu_i$ represents the mean effect of the $i$-th treatment and the errors $\epsilon_{ij}$'s are independent and normally distributed with
$\epsilon_{ij} \sim N(0, \sigma_i^2)$ for $j = 1,2,\ldots,n_i; i = 0,1,\ldots,k$.

\indent Researchers may have prior knowledge about the specified ordered structure for the mean effects of various treatments, e.g., tree order, umbrella order, simple loop structure, etc. Utilizing the information on order restrictions on parameters leads to more powerful tests. One may see \cite{brunk1972statistical} and \cite{dykstra2012advances} for detailed studies on inference procedures in ANOVA models under various order constraints. A recent application of tree ordering in drug testing is discussed in \cite{peddada2006tests}. In this study, the haematopoietic data from the National Toxicology Program's (NTP) prechronic investigation of anthraquinone in male rats is examined. 
From the comprehensive blood count data, mean corpuscular volume (MCV) was selected as a response variable. The impact of six dosages $(0, 1875, 3750, 7500, 15,000, \;\text{and}\; 30,000 \; ppm)$ of a particular chemical group called anthraquinone was measured in male rats, and the data was taken for day 22. It was observed that the mean treatment effects satisfy a tree order. Another application of tree order for treatment effects is discussed in Sect. \ref{Application}. Here, we have considered the efficacy of various psychological treatments on the noise sensitivity of individuals who suffer from headaches.

For a one-way ANOVA model, \cite{bartholomew1961test} developed the likelihood ratio test (LRT) for testing against tree ordered alternatives when group variances are known. \cite{robertson1985one} considered this problem when variances are unknown but equal. They computed the critical points of the LRT for selected values of $k$ and sample sizes.  An approximation to the quantiles of this test was derived by \cite{futschik1998likelihood} for a large number of treatments.

The admissibility properties of LRT and other multiple contrast tests for tree ordered alternatives have been studied in \cite{cohen1992improved}, \cite{cohen2000properties} and \cite{cohen2002inference} when the group variances are equal. A new test statistic based on the restricted maximum likelihood estimators (MLEs) of the means was proposed by \cite{marcus1992further}. For some specific choices of sample sizes and number of groups, they gave numerical schemes to compute the null distribution and simulate the power values of this test. \cite{singh1993power} have numerically computed the power values of the LRT when the number of groups is 3 or 4, and error variances are common but may be known or unknown.  

Stochastic domination properties of restricted MLEs of the mean effects and confidence intervals based on these are investigated in \cite{hwang1994confidence}. A test based on the differences of the restricted MLEs of means was proposed by \cite{peddada2006tests} for tree ordered alternatives. \cite{betcher2009statistical} developed similar test procedures based on a modification of the restricted MLEs. This is due to the fact that the performance of the restricted MLE for the mean of the control group drastically drops as the number of treatment groups increases (\cite{lee1988quadratic}).

\cite{tian2003likelihood} considered the problem of testing the homogeneity of several inverse Gaussian means against alternatives of monotone ordering and tree ordering. They developed the LRT for the two cases when the common scale parameter is known or unknown. \cite{chaudhuri2007consistent} derived a consistent estimator of the mean of the control when the sample size for the control group is large. Two new estimators for this were also given in \cite{momeni2021estimation}, and they are shown to have comparable good risk performance. 

In all of the works described above, the error variances are assumed to be equal for all the groups, known or unknown. However, in various industrial or medical studies, it is not possible to have common variance due to varying experimental conditions. It has been noted that for such situations, the classical F-test in ANOVA becomes either too conservative or liberal; see \cite{bishop1978exact}, \cite{chen2000one}, and \cite{pauly2015asymptotic}.

\cite{chen2001one} developed one-stage and two-stage sampling procedures for testing the equality of means in heteroscedastic one-way ANOVA. Also, they proposed a new range test and ANOVA test based on  Student’s t distribution for one-stage and two-stage sampling procedures and calculated critical values. \cite{chen2000one} developed similar one-stage range test procedures for testing equality means against simple ordered means. \cite{chen2004one} modified the test given by \cite{chen2000one} for simple ordered alternative and compares their results. \cite{wen2022single} used the stage sampling procedure to obtain a one-sided and two-sided confidence interval of $\mu_i-\mu_0$ $\forall i=1,2,\ldots,k$.

\cite{tamhane2004finding} considered the problem of identifying the maximum safe dose level of a drug when group variances are heterogeneous. They proposed three test procedures and developed approximation methods to find the critical points. 

\cite{hasler2008multiple} compared four different tests for testing multiple contrasts when errors are normally distributed and have heteroscedastic variances. A simulation study was performed to check whether these tests control type I error rates. It was noted that depending on choices of group variances, the tests may be liberal or conservative. \cite{hasler2016heteroscedasticity} further compared one of the tests studied in \cite{hasler2008multiple} with another one based on the sandwich estimation method of \cite{herberich2010robust} for multiple contrasts with respect to controlling family-wise error rates and found the former method to be more robust. 

\cite{li2012multiple} considered the multiple comparisons with the control problem and proposed three techniques to obtain approximate confidence intervals, using Bonferroni approximation, Slepian’s inequality, and multivariate t distribution. Also, they simulated the size for this test and compared it with other alternative tests. For the same problem, an alternative test based on multiplicity adjustment to the Dunnett procedure was proposed by \cite{tamhane2023multiplicity}. They have proposed an algorithm to estimate the size of this test and compared it to similar existing tests. \cite{alver2023multiple} proposed a parametric bootstrap (PB) test for testing the hypothesis that at least one treatment effect differs from the control. 

Note that in the studies described above, the main focus is on proposing multiple comparison-based tests that use asymptotic or approximate approaches to find critical points under heteroscedastic variances. However, the null and alternative hypotheses are formulated differently from those in our study. 
Here, we test the hypothesis of homogeneity of treatment effects against the alternative of treatments better than the control with at least one strict inequality. That is testing against the alternative of at least one treatment effect being superior to the control, thereby making it a constrained testing problem.

In this paper, we consider testing for at least one treatment better than control in a one-way ANOVA model when group variances are heterogeneous. The likelihood ratio test and simultaneous comparison-based tests are developed, and parametric bootstrap (PB) is used to evaluate critical points. An extensive simulation study shows that the proposed tests behave as exact tests even for small samples as well as for highly heterogeneous groups. Power comparisons capture three distinct situations, each illustrating a particular configuration where one test is seen to outperform the other two. Recently, \cite{alver2023multiple} have proposed a parametric bootstrap  test based on multiple comparisons for testing the homogeneity of effects against the alternative of at least one effect being different than the control. A comparison of our proposed tests with the test of \cite{alver2023multiple} (called PB-AZ henceforth) is also carried out. Our tests (mainly LRT and Max-D) are seen to exhibit superior performance in controlling type-I error rates for small sample size cases and have higher powers over PB-AZ for almost all cases of parameter configurations.


The paper is arranged in the following manner. Sect. \ref{LRT} describes the likelihood ratio test procedure. The algorithm for finding the MLEs of the parameters for the control problem is developed, and the procedure to determine the critical values of the LRT, as well as proof of the accuracy of the PB method, is given in this section. In Sect. \ref{Bootstrap Max-D Min-D}, two tests based on pairwise differences of group means are proposed, and the PB approach is used to determine their critical values. The asymptotic accuracy of these tests is established in Sect. \ref{Section 3.1}. An extensive simulation study is conducted in Sect. \ref{Size and Power} for size and power comparisons of our proposed tests and PB-AZ test. Details of the `R' package developed for the practical application of the tests are also given here.  Further, the robustness of these tests is also investigated under deviation from normality in Sect. \ref{Robustness}. Finally, in Sect. \ref{Application}, the test procedures are implemented on a data set of patients undergoing different psychological treatments. 

\section{Likelihood Ratio Test}\label{LRT}
In this section, we develop the likelihood ratio test for testing $H_0$ against $H_1$ as defined in (\ref{eq:1.1}). Note that here, we have assumed that the variances are unknown and heterogeneous. 
It is required to maximize the likelihood function under the null and full parameter spaces. These are described as 
$\Omega_{H_0} = \{(\mu_0,\mu_1,\ldots,\mu_k),\;(\sigma_0^2,\sigma_1^2,\ldots,\sigma_k^2):\;\sigma_i^2>0,\;\mu_i=\mu,\;\forall\; i=0,1,\ldots,k\}$ and $\Omega_{T} = \{(\mu_0,\mu_1,\ldots,\mu_k),\;(\sigma_0^2,\sigma_1^2,\ldots,\sigma_k^2):\;-\infty<\mu_0\leq \mu_i<\infty,\;i=1,2,\ldots,k,\;\sigma_j^2>0,\;j=0,1,\ldots,k\}$ respectively. The likelihood ratio is then defined as $\lambda_l=\frac{\hat{L}\left(\Omega_{H_0}\right)}{\hat{L}(\Omega_T)}$, where $\hat{L}\left(\Omega_{H_0}\right)$ and $\hat{L}(\Omega_T)$ represents the maximum values of the likelihood function ($L$) under $\Omega_{H_0}$ and $\Omega_T$ respectively. 

It is noted that \cite{mondal2023testing} have given an algorithm to determine the MLEs of  ${\mu}$ and ${\sigma}_{i}^{2}$ under the null hypothesis parameter space $\Omega_{H_0}$. Using this algorithm, the MLEs of ${\mu}$ and ${\sigma}_{i}^{2}, i=0, 1, \ldots, k$ are obtained as $\hat{\mu}_{(0)}$ and $\;\hat{\sigma}_{i (0)}^{2}, i=0, 1, \ldots, k$ respectively. Next, we propose a method for evaluating the MLEs of the parameters under $\Omega_{T}$. 

Denoting the observed values of random variables $X_{ij}$ by $x_{ij}, j=1, \ldots, n_i; i=0, 1, \ldots, k$, the likelihood function under $\Omega_{T}$ is given by
$$
{L}\left(\Omega_{T}\right)=\frac{1}{(2 \pi)^{N / 2}} \frac{1}{\prod_{i=0}^{k}\left({\sigma}_{i}^2\right)^{n_{i} / 2}} \exp \left(-\frac{1}{2} \sum_{i=0}^{k} \sum_{j=1}^{n_{i}} \frac{\left(x_{i j}-{\mu_i}\right)^{2}}{{\sigma}_{i}^{2}}\right).
$$

The log-likelihood function can be written as 
\begin{equation}
    l\left({\underline {\boldsymbol{\mu}}}, \boldsymbol {\underline \sigma^{2}}\right)=-\sum_{i=0}^{k}\frac{n_{i}}{2} \ln \sigma_{i}^{2}-\sum_{i=0}^{k}\frac{1}{2 \sigma_{i}^{2}} \sum_{j=1}^{n_{i}}\left(x_{i j}-\mu_{i}\right)^{2}+\delta ,\label{2.1}
\end{equation} 
where $\boldsymbol{\underline{\mu}}= (\mu_0,\mu_1,\ldots,\mu_k)^T$, $\boldsymbol{\underline{\sigma}}^2= (\sigma_0^2,\sigma_1^2,\ldots,\sigma_k^2)^T$ and $\delta$ is a constant (independent of parameters). 

Let $\overline{x}_{i}=\frac{1}{n_i}\sum\limits_{j=1}^{n_i}x_{ij}, ~i=0, 1\ldots,k$. The MLEs are solutions of the likelihood equations 

\begin{equation}
    {\mu}_i= \overline{x}_{i}\;\text{and}\;
{\sigma}_i^{2}=\frac{1}{n_i}\sum\limits_{j=1}^{n_i}\left(x_{ij}-{{\mu}}_{i}\right)^2, ~~i=0, 1, \ldots, k \label{2.2}
\end{equation}
subject to the condition that $\left({\underline {\boldsymbol{\mu}}}, \boldsymbol {\underline \sigma^{2}}\right) \in \Omega_T$. 

Note that the likelihood equations (\ref{2.2}) for $(2k+2)$ parameters are nonlinear in nature, and one cannot solve these analytically under constraints $\Omega_T$. A nonlinear Gauss-Seidel type iterative method is proposed here to solve this optimization problem. The method uses the fact that for fixed $\sigma_i^2$s, the maximization of $l\left({\underline {\boldsymbol{\mu}}}, \boldsymbol {\underline \sigma^{2}}\right)$ in equation (\ref{2.1}) is equivalent to minimizing \begin{equation}
    \sum_{i=0}^{k}\left(\overline{x}_{i}-\mu_{i}\right)^{2} w_{i} \label{2.3}
\end{equation}
with respect to $\underline{\boldsymbol{\mu}} \in I$, where $I = \{\underline{\boldsymbol{\mu}}:\mu_0\leq \mu_i, i=1,2,\ldots k\}$, and $w_{i}=n_{i} / \sigma_{i}^{2}$ for $i=0,1,\ldots ,k$. A solution to this can be obtained using 
the Minimum Violator Algorithm (see Chapter 2, \cite{brunk1972statistical}). We use these solutions to compute the next iteration for $\sigma_i^2$s. The detailed algorithm is described below.

Define $\overline{\underline{\bf x}}=(\overline{x}_{0},\overline{x}_{1},\ldots,\overline{x}_{k})^T, {\underline {\bf w}}=(w_0,w_1,\ldots, w_k)^T$ and $s_i^2=\frac{1}{n_i}\sum_{j=1}^{n_i} {(x_{ij}-\bar{x}_i)^2}; i=0, 1, \ldots, k$. Further, let $\hat{\mu}_{i(T)}$ and $\;\hat{\sigma}_{i(T)}^{2}$ be the MLEs of ${\mu_i}$ and ${\sigma}_{i}^{2}$ under $\Omega_T$.


\begin{algorithm}[!ht]
    \caption{Scheme for calculating the MLE under tree order restrictions on mean effects}\label{Algorithm 1}
{Step 1:} Consider the initial values as $\underline{\boldsymbol{\mu}}^{(0)}=\underline{\overline{\mathbf{x}}}$ and ${\sigma}_i^{2(0)}=\frac{1}{n_i}\sum\limits_{j=1}^{n_i}\left(x_{ij}-{{\mu}}^{(0)}_{i}\right)^2$, $i=0,1,\ldots,k$, where $\underline{\boldsymbol{\mu}}^{(0)}=\left(\mu_0^{(0)},\mu_1^{(0)},\ldots,\mu_k^{(0)}\right)^T$.

{Step 2:} Use the Minimum Violator Algorithm to compute the isotonic regression $\underline{\boldsymbol{\mu}}^{(m)}$ of $\left(\underline{\overline{\mathbf{x}}}, \mathbf{\underline{w}}^{(m-1)}\right)$ under tree order restriction, where  $\underline{\boldsymbol{\mu}}^{(m)}=\left(\mu_0^{(m)},\mu_1^{(m)},\ldots,\mu_k^{(m)}\right)^T$, $\mathbf{\underline{w}}^{(m-1)}=\left({w}_0^{(m-1)},{w}_1^{(m-1)},\ldots,{w}_k^{(m-1)}\right)^T$, ${w}_i^{(m-1)}=\frac{n_i}{{\sigma}_i^{2(m-1)}}$; $i=0,1,\ldots ,k$.\\

{Step 3:} Now calculate ${\sigma}_i^{2(m)}=\frac{1}{n_i}\sum\limits_{j=1}^{n_i}\left(x_{ij}-{{\mu}}^{(m)}_{i}\right)^2,$ and obtain \\$\underline{\boldsymbol{\sigma}}^{2(m)}=\left({\sigma}_0^{2(m)},{\sigma}_1^{2(m)},\ldots, {\sigma}_k^{2(m)}\right)^T.$\\
{Step 4:} Terminate the algorithm in Step $3$ if 
$$
\max _{0 \leq i \leq k}\left|\mu_{i}^{(m-1)}-\mu_{i}^{(m)}\right| \leq 10^{-p}
$$
and $$\max _{0 \leq i \leq k}\left|\sigma_{i}^{2(m-1)}-\sigma_{i}^{2(m)}\right| \leq 10^{-p}
\;\;\text{for some $p\geq 3.$}$$
\end{algorithm}

Note that Algorithm \ref{Algorithm 1} is developed based on the fact that   
$l\left(\underline{\boldsymbol{\mu}}^{(m)},\underline{\boldsymbol{\sigma}}^{2(m-1)}\right)$ maximizes $l\left(\underline{\boldsymbol{\mu}}, \underline{\boldsymbol{\sigma}}^{2(m-1)}\right)$ over $\underline{\boldsymbol{\mu}} \in I$. Therefore, we have
\begin{equation}
l\left(\underline{\boldsymbol{\mu}}^{(m)}, \underline{\boldsymbol{\sigma}}^{2(m-1)}\right) \leq l\left(\underline{\boldsymbol{\mu}}^{(m)}, \underline{\boldsymbol{\sigma}}^{2(m)}\right) \leq l\left(\underline{\boldsymbol{\mu}}^{(m+1)}, \underline{\boldsymbol{\sigma}}^{2(m)}\right).\label{2.4}
\end{equation}


The existence and uniqueness of the MLE for $\left({\underline {\boldsymbol{\mu}}}, \boldsymbol {\underline \sigma^{2}}\right)$ subject to tree order restrictions on $\mu_i$'s can be established using arguments similar to those in \cite{shi1998maximum}. To prove the uniqueness, the following condition is required\\
\textbf{Condition 1}: ${s}_{i}^{2}> \max \left\{\left(\overline{x}_{i}-a\right)^{2},\left(\overline{x}_{i}-b\right)^{2}\right\},  \;i=0, 1, \ldots, k,$ 
where $a=\min\limits_{0\leq i \leq k} \overline{x}_i,\;\; b=\max\limits_{0\leq i \leq k} \overline{x}_i$.\\
Finally, the convergence of Algorithm \ref{Algorithm 1} is shown in Theorem \ref{Theorem 1}. The proof is given in the Appendix.

\begin{theorem}
\label{Theorem 1}
Under {\bf Condition 1}, the sequence of iterative solutions in Algorithm \ref{Algorithm 1} converges to the true MLE of $(\boldsymbol{\underline\mu},\boldsymbol{\underline\sigma}^2)$. 
\end{theorem}

Now, we obtain the likelihood ratio as \begin{equation}
    \lambda_l=\prod_{i=0}^{k} \frac{\left(\hat{\sigma}_{i (T)}^{2}\right)^{n_{i} / 2}}{\left(\hat{\sigma}_{i (0)}^{2}\right)^{n_{i} / 2}}.\label{2.5}
\end{equation}
Then $H_0$ is rejected if $\lambda_l<c_{l}$, where $c_l$ is determined from the size condition $\sup\limits_{H_0}P{(\lambda_l<c_{l})}=\alpha$. Note that the expressions for $\lambda_l$ and its null distribution can not be obtained in closed forms, so we use the PB approach for computing the critical points. Note that the distribution of $\lambda_l$ does not depend on the location parameters under $H_0$. So, the bootstrap samples $X_{ij}^*, \;j= 1,2,\ldots,n_i;$ are generated from $N(0,{S}_{i}^2),$ ${S}_{i}^2=\frac{1}{n_i-1}\sum\limits_{j=1}^{n_i}\left(X_{ij}-\overline{X}_{i}\right)^2,$ $\overline{X}_{i}=\frac{1}{n_i}\sum\limits_{j=1}^{n_i} X_{ij}, i=0,1,\ldots,k.$ 

The estimated values of the critical points using the bootstrap and empirical size and power values of the LRT are now evaluated using the following algorithm. 
\begin{algorithm}[!ht]
    \caption{Method to evaluate critical points, empirical size, and power values of the LRT}
    \label{Algorithm 2}
{Step 1.} Choose configurations $\mu$, $(\sigma_0^2, \sigma_1^2, \ldots \sigma_k^2)$.
    
{Step 2.} Generate samples $X_{ij},$ $j=1,2,\ldots,n_i$  from $N(\mu, \sigma_i^2)$, $i=0,1,\ldots,k$.

{Step 3.} Compute the sample means $\overline{X}_i$, sample variances $S_i^2, i=0, 1, \ldots, k$. 

{Step 4.} Generate bootstrap samples $X_{ij}^*, j=1, \ldots, n_i$  from $N(0, {S}_{i}^2), i=0,1,\ldots,k.$

{Step 5.} Compute the bootstrap likelihood ratio $\lambda_l^*$ using the samples $X_{ij}^*$. 

{Step 6.} Repeat Steps 4 and 5 $M$ number of times (say $M=5000$), and arrange the corresponding bootstrap test statistic values in ascending order to determine the $\alpha$-quantile $c_{l}^*:=\lambda_{l(\lfloor{\alpha M}\rfloor)}^*$ as the bootstrap critical value.

{Step 7.} Evaluate the likelihood ratio $\lambda_l$ (defined in (\ref{2.5})) using $X_{ij}$ for $j=1,\ldots ,n_i$; $i=0,1,\ldots ,k$. 

{Step 8.} Steps 2 to 7 are repeated $P$ number of times (say, $P=5000$). Call the computed values of the test statistic in Step 7 as $\lambda_{l\gamma}$ and that of the bootstrap critical point in Step 6 as $c_{l\gamma}^*$ for $\gamma=1,2,\ldots, P.$

{Step 9.} The empirical size value of the LRT is then $$\alpha_{LRT} = \frac{\text{total count of}\; \lambda_{l\gamma} < c_{l\gamma}^*}{P}=\frac{1}{P}\sum\limits_{\gamma=1}^P {\bf 1}\{\lambda_{l\gamma} < c_{l\gamma}^*\}$$ 

{Step 10.} For evaluating the power values, change the configuration in Step 1 by $(\mu_0, \mu_1, \ldots, \mu_k)$ with $\mu_0\leq \mu_i, i=1, \ldots, k$ (with at least one strict inequality) and generate samples from $N(\mu_i, \sigma_i^2),$ $i=0,1,\ldots,k$.  
\end{algorithm}

To validate the bootstrap test, it is necessary to show that the conditional distribution of the bootstrap LRT, given the samples, converges to the null distribution of the LRT. The following theorem proves the asymptotic accuracy of the parametric bootstrap used here.
\begin{theorem}
\label{Theorem 2}
If $\min\limits_{0 \leq i \leq k} n_i \rightarrow \infty$ then 
\[
\sup_{x \in \mathbb{R}}|F_{\lambda_l^*|\underline{\mathbf{X}}}(x) - F_{\lambda_l|H_0}(x)| \stackrel{P}{\longrightarrow} 0,
\]
where $F_{\lambda_l^*|\underline{\mathbf{X}}}(x)$ represents the conditional distribution function of $\lambda_l^*$ given the observations, and $F_{\lambda_l|H_0}(x)$ represents the null distribution function of $\lambda_l$.
\end{theorem}
\begin{proof}
The proof follows steps similar to those given in \cite{mondal2023testing}.
\end{proof}

\cite{lee1988quadratic} showed that the mean square error of the restricted MLE of the mean effect of the control group exceeds that of the unrestricted MLE as the number of treatments increases. As a consequence, the performance of the LRT may be affected. Therefore, in the next section, we propose two new test procedures based on pairwise standardized differences of sample means.

\section{Bootstrap Based Max-D and Min-D Tests}\label{Bootstrap Max-D Min-D}

 Let    
    $D_i=\frac{\overline{X}_{i}-\overline{X}_0}{\sqrt{\frac{S^2_{i}}{n_{i}}+\frac{S^2_{0}}{n_{0}}}},\;i=1,2,\cdots k$ and ${\bf\underline {D}}=(D_1, D_2,\cdots, D_{k})$. 

    Consider the test statistics  \begin{equation}
        D_{Max} = \text{max}\;{\bf{\underline {D}}}, \;\;\text{and}\;\;
        D_{Min} = \text{min}\;\bf{\underline {D}}.\label{3.1}
    \end{equation}

We reject the null hypothesis $H_0$ using $D_{Max}$ if $D_{Max}>d_{Max}$ and using $D_{Min}$ if $D_{Min}>d_{Min}$. Here, critical points $d_{Max}$ and $d_{Min}$ are determined so that the tests have size $\alpha$. The null distributions of the test statistics cannot be obtained in closed forms, so we use the parametric bootstrap to find the critical points.

In practical applications, it is noted that there may be two types of situations with respect to the alternative hypothesis $H_1$. For example, in comparing the efficacy of various drugs for a certain disease, the control group may be given a placebo. In this case, there may be strict inequality for all $i$ in $H_1$ as patients not being treated with any drug may have inferior outcomes in terms of survival rates, time to cure, etc. In comparing the breaking strengths of the front parts of Formula One cars, some cars may use materials with similar breaking strengths, so there may be some equalities in $H_1$. We see in simulation studies that Min-D performs better in terms of power values in the first situation case and Max-D in the second case.  

We have also considered modified versions of Min-D and Max-D by replacing treatment means with the restricted MLEs and standardizing using standard errors of these differences. However, it was observed in simulation studies that the empirical size and power values of these tests are almost the same as those for Min-D and Max-D. This may be due to the earlier result of \cite{lee1988quadratic} that the restricted MLE of the control group effect performs worse than the sample mean when number of groups is large. Further, the computational complexity of these modified test procedures is at least ten times that of Min-D and Max-D when the number of samples in the simulation and bootstrap steps is moderate. Computation time increases exponentially as number of samples increases. Hence, we have not included these tables in our studies.

The procedure for determining the critical points, empirical size, and power values for Max-D and Min-D tests is described in Algorithm \ref{Algorithm 3}.







\begin{algorithm}
    \caption{Method to evaluate size and power of Max-D and Min-D tests}\label{Algorithm 3}

{Step 1.} Perform Steps 1-3 from Algorithm \ref{Algorithm 2}.

{Step 2.} Now, generate bootstrap samples ($X_{ij}^*$) from normal distribution with mean 0 and variance $S_i^2$ and calculate bootstrap test statistics  $$D_{Max}^* = \text{max}\;{\bf{\underline {D}^*}}, \;\;\text{and}\;\; \\ D_{Min}^* = \text{min}\;\bf{\underline {D}^*},$$ where ${\bf\underline {D}^*}=(D_1^*, D_2^*,\cdots, D_{k}^*),\; D_i^*=\frac{{\overline{X}}{^*_i}-{\overline{X}}{^*_0}}{\sqrt{\frac{S^{*2}_{i}}{n_{i}}+\frac{S^{*2}_{0}}{n_{0}}}},\;(i=1,2,\cdots k)$.

{Step 3.} Repeat Step 2 many times, say, $K = 5000$. Subsequently, we obtain 5000 $D_{Max}^*$ and $D_{Min}^*$ values. Then arrange K values of $D_{Max} ^*$ and $D_{Min}^*$ in non-decreasing order, say $D_{Max(1)}^*\leq D_{Max(2)}^*\leq \cdots \leq D_{Max(K)}^*$ and $D_{Min(1)}^*\leq D_{Min(2)}^*\leq \cdots \leq D_{Min(K)}^*$. 

{Step 4.} Calculate the bootstrap quantile $d_{Max}^*$, $d_{Min}^*$ as the $1-\alpha$ quantile. Then reject $H_0$ at $\alpha$ significance level if $D_{Max}>d_{Max}^*$ and $D_{Min}>d_{Min}^*$. In other words, $$d_{Max}^*:=D_{Max([(1-\alpha)K])}^* \;\; \text{and} \;\; d_{Min}^*:=D_{Min([(1-\alpha)K])}^*.$$

{Step 5.} Using data from Step 1, calculate test statistics, $D_{Max}$ and $D_{Min}$.

{Step 6.} Now repeat Steps 1 to 5 for a large number of times, say $L = 5000,$ obtain L number of values of $d_{Max}^*,$ $d_{Min}^*,$ $D_{Max}$ and $D_{Min}$ say $d_{Max_r}^*,$ $d_{Min_r}^*,$ $D_{Max_r}$ and $D_{Min_r}, r=1,2,\cdots L.$ 

{Step 7.} Finally, one can get the estimated value of sizes for Max-D and Min-D as $$\alpha_{Max} = \frac{\text{number of times}\; D_{Max_r} > d_{Max_r}^*}{L} \;\;\text{and} $$$$ \alpha_{Min} = \frac{\text{number of times}\; D_{Min_r} > d_{Min_r}^*}{L}.$$

{Step 10.} To calculate power, we use $\mu_0\leq \mu_i$ (with a strict inequality in Step 1) and follow the same procedure.
\end{algorithm}

Using the critical point for Max-D statistic, $(1-\alpha)$ one sided simultaneous confidence intervals for $\mu_i-\mu_0,~~ i=1, \ldots, k$ are given by 
$\left(\overline{X}_i - \overline{X}_0 - d_{Max}^* \sqrt{\frac{S^2_{i}}{n_{i}} + \frac{S^2_{0}}{n_{0}}}, \infty\right),~~ i=1, \ldots, k$.

\subsection{Asymptotic Accuracy}\label{Section 3.1}
In this section, we prove the asymptotic accuracy of parametric bootstrap-based Max-D and Min-D tests under the asymptotic setup
\begin{equation}
    \frac{n_i}{N}\rightarrow p_i \;{\text as}\; \underset{0 \leq i \leq k}{\min} n_i \rightarrow \infty\; {\text where}\; i=0,1,\cdots k, \;N= \sum\limits_{i=0}^{k}n_i.\label{3.3}
\end{equation}

The following theorem demonstrates the asymptotic precision of the bootstrap-based Max-D and Min-D tests.
\begin{theorem}
\label{Theorem 4}
    Under the asymptotic setup (\ref{3.3}) $$\sup\limits_{\underline t \in \mathbb{R}^k}|F_{\underline{{\mathbf{D}}}^*|\underline{\mathbf{X}}}(\underline t)-F_{\underline{{\mathbf{D}}}|H_0}(\underline t)|\stackrel{P}{\longrightarrow}0,$$
    where $F_{\underline{{\mathbf{D}}}^*|\underline{\mathbf{X}}}(\underline t)$ denote the conditional distribution function of $\mathbf{D}^*$ given observations $\underline{\mathbf{X}}$ and $F_{\underline{{\mathbf{D}}}|H_0}(\underline t)$ denote the null distribution function of  $\mathbf{D}.$
\end{theorem}

\begin{proof}
    Given samples $X_{ij}\sim N(\mu_i, \sigma_i^2), i = 0,1,\ldots, k.$ Under $H_0,$ $\mu_i=\mu$ for all $i$ (say). We can take $\mu=0$  without loss of generality. Now under $H_0$
    $$\underline{\bf V} = (\sqrt{N}{\overline{X}}{_0}\;, \sqrt{N}{\overline{X}}{_1}\;, \ldots 
    \sqrt{N}{\overline{X}}{_k})^T \sim N_{k+1}(\underline{\bf 0}, \boldsymbol{\Sigma_V})$$
    and
$$\underline{\bf V}^* = (\sqrt{N}{\overline{X}}{^*_0}\;, \sqrt{N}{\overline{X}}{^*_1}\;, \ldots 
    \sqrt{N}{\overline{X}}{^*_k})^T \sim N_{k+1}(\underline{\bf 0}, \boldsymbol{\hat{\Sigma}_V}),\; $$
    where
$\boldsymbol{\Sigma_V} = \text{diag}\left(\frac{N\sigma_0^2}{n_0}, \frac{N\sigma_1^2}{n_1}, \cdots, \frac{N\sigma_k^2}{n_k}\right)$ and $\boldsymbol{\hat{\Sigma}_V} = \text{diag}\left(\frac{NS_0^{2}}{n_0}, \frac{NS_1^{2}}{n_1}, \cdots, \frac{NS_k^{2}}{n_k}\right).$
\\
Utilizing the Kullback-Leibler divergence of $\underline{\bf V}^*$ from $\underline{\bf V}$, we obtained 
\begin{equation}
    D_{KL}(\underline{\bf V}^*||\underline{\bf V}) = \frac{1}{2}\left[\sum\limits_{i=0}^k\log\left(\frac{\sigma_i^2}{S_i^{2}} \right) +\sum\limits_{i=0}^k\frac{S_i^{2}}{\sigma_i^2} -(k+1) \right]\stackrel{P}{\longrightarrow}0 \;\;\text{as} \;\; \underset{0 \leq i \leq k}{\min} n_i \rightarrow \infty.\label{3.4}
\end{equation}

Define ${\bf B}=\begin{bmatrix}
  -1 & 1 & 0 &\cdots & 0 \\
  -1 & 0 & 1 &\cdots & 0 \\
  \vdots & \vdots & \vdots & \ddots & \vdots \\
  -1 & 0 & 0 &\cdots & 1 \\
\end{bmatrix}$,
$Z_i = \sqrt{N}({\overline{X}}{_i}-{\overline{X}}{_0}),$ $Z_i^* = \sqrt{N}({\overline{X}}{^*_i}-{\overline{X}}{^*_0}),$ $i=1,2,\ldots,k.$
Then we can write $$\underline{\bf Z} = (Z_1,Z_2,\cdots,Z_k)^T={\bf B}\underline{\bf V}\;\;{\text{and}}\;\;\underline{\bf Z}^* = (Z_1^*,Z_2^*,\cdots,Z_k^*)^T={\bf B}\underline{\bf V}^*$$
From equation (\ref{3.4}) we get $D_{KL}(\underline{\bf Z}^*||\underline{\bf Z})\stackrel{P}{\longrightarrow}0$ as $\underset{0 \leq i \leq k}{\min} n_i \rightarrow \infty.$ This implies \begin{equation}
    \sup\limits_{\underline z \in \mathbb{R}^k}|F_{\underline{{\mathbf{Z}}}^*|\underline{\mathbf{X}}}(\underline z)-F_{\underline{{\mathbf{Z
    }}}|H_0}(\underline z)|\stackrel{P}{\longrightarrow}0, \label{3.5}
\end{equation}

Consider $\hat{\bf Q}= \text{diag}\left(\frac{\sqrt{n_1n_0}}{N\hat{c}_1}, \frac{\sqrt{n_2n_0}}{N\hat{c}_2}, \ldots, \frac{\sqrt{n_kn_0}}{N\hat{c}_k}\right)$, $\hat c_i^2=\frac{n_0}{N}S_i^2 + \frac{n_i}{N} S_0^2,$ and $\hat{\bf Q}^*=\text{diag}\left(\frac{\sqrt{n_1n_0}}{N}\frac{1}{\hat{c}_1^*}, \frac{\sqrt{n_2n_0}}{N}\frac{1}{\hat{c}_2^*}, \ldots, \frac{\sqrt{n_kn_0}}{N}\frac{1}{\hat{c}_k^*}\right)$, $\hat{c}^{*2}_i={\frac{n_i}{N}S^{*2}_{0}+\frac{n_0}{N}S^{*2}_{i}},$ $i=1,2,\cdots,k,$
then $\underline{\bf D}$ and $\underline{\bf D}^*$ can be written as $\underline{\bf D} = \hat{\bf Q}{\underline{{\mathbf{Z}}}}$ and $\underline{\bf D}^* = \hat{\bf Q}^*{\underline{{\mathbf{Z}}}}^*$.

Also using the fact that as the number of observations $n_i$ approaches infinity for $i=0,1,\cdots, k$,  $S_i^{*2}\stackrel{P}{\longrightarrow} \sigma_i^2$ in conditional probability, given observations $\underline{\bf X}$ showed by \cite{bickel1981some}. Therefore we obtained $\hat{c}_i^{*2}\stackrel{P}{\longrightarrow}c_i^2={p_i\sigma^{2}_{0}+p_0\sigma^{2}_{i}},\;i=1,2,\ldots,k$. Also it is easy to observe that $\hat{c}_i^{2}\stackrel{P}{\longrightarrow}c_i^2,\;i=1,2,\ldots,k$.

 Under (\ref{3.3}), it is clear that \begin{equation}
     \hat{\bf Q}\stackrel{P}{\longrightarrow}{\bf Q} \;\;\text{and}\;\; \hat{\bf Q}^*\stackrel{P}{\longrightarrow}{\bf Q},\label{3.6}
 \end{equation}
 where ${\bf Q} = \text{diag}\left(\frac{{\sqrt{p_1p_0}}}{{c}_1}, \frac{\sqrt{p_2p_0}}{{c}_2}, \ldots, \frac{\sqrt{p_kp_0}}{{c}_k}\right)$.
By using the multivariate Slutsky theorem and utilizing equations (\ref{3.5}) and (\ref{3.6}), we obtain the desired result.
\end{proof}

\section{Size and Power Analysis}\label{Size and Power}
Here, we present an extensive simulation study about the size and power values of the three tests, Max-D, Min-D, and LRT, for various configurations of sample sizes and variances. Empirical size values have been computed for a nominal size $\alpha = 0.05.$ Five thousand random samples $X_{ij},~ j = 1, 2,\ldots, n_i $ are generated from normal populations with a mean $\mu$ and variance $\sigma_i^2, ~i=0,1,\ldots,k$. 
For each sample, 5000 bootstrap samples are generated from $N(0, S_i^2), ~i=0, \ldots, k$ populations.  
Algorithms \ref{Algorithm 2} and  \ref{Algorithm 3} are used to compute the critical points, empirical sizes, and powers of the three tests. For computing critical points and size values, $\mu=0$ is taken as the null distributions of test statistics are independent of the location.  For power calculations, we choose the configurations for $\mu_i$'s as $c(\mu_0, \mu_1,\ldots,\mu_k)$ for $k = 3, 4,$ where $c$ ranges from 1 to 6.1 with an increment of 0.3. This type of choice of the treatment means is selected to observe the effect on powers of tests, as the differences between the values of $\mu_i$ with those of $\mu_0$, increase.


The empirical sizes of all three tests are reported in Tables \ref{Table 1} -  \ref{Table 3} for $k=2,3,4$ by taking small with equal, small with unequal, and moderate with unequal treatment sample sizes with homogeneous, heterogeneous, and highly heterogeneous variances configurations. Further, we provide power curve plots for $k=2,3$ by considering different mean vectors $(\underline{\boldsymbol\mu}=c(1, 1.3, 1.6), c(1, 1.3, 1), c(1, 1.3, 1.3), c(1, 1.6, 1.9, 1.3),\\ c(1, 1.6, 1.3, 1), c(1, 1.6, 1.6, 1.6))$ with four specific cases of other parameters given in Table \ref{Table 0}.

\begin{table}[!ht]
    \caption{Configurations of parameters for size and power calculations}
    \label{Table 0}
    \centering \setlength\tabcolsep{1pt}
    \begin{footnotesize}
    \begin{tabular}{c | c}
    \hline
        \textbf{Sample Sizes} $({N})$ & \textbf{Variances} ($\boldsymbol{\underline{\sigma}^2}$)\\
        \hline
        \multicolumn{2}{c}{\textbf{Empirical size configurations}} \\
        \hline
        $N_1(5,5,5)$, $N_2(20,10,25)$, & Homogeneous\\
        $N_3(5,8,12,10)$, $N_4(20,15,35,25)$,  & Moderate heterogeneous \\
         $N_5(10,10,10,10,10)$, $N_6(5,5,5,20,15)$ & Highly heterogeneous \\
                                                                   & (Given in Table \ref{Table 1} - \ref{Table 3})\\
        \hline
        \multicolumn{2}{c}{\textbf{Power curve configurations}} \\
        \hline
        (a) {Small sample Sizes:} $N_1(5,5,5)$, $N_5(5,8,12,10)$ & Moderate heterogeneous: $(2,3,4)$, $(2,3,4,5)$\\
        \hline
        (b) {Moderate sample Sizes:} $N_1(20,10,25)$, $N_5(20,15,35,25)$ & Moderate heterogeneous: $(2,3,4)$, $(2,3,4,5)$\\
        \hline
        (c) {Small Sample Sizes:} $N_1(5,5,5)$, $N_5(5,8,12,10)$ & High heterogeneous: $(3,1,6)$, $(3,1,6,4)$\\
        \hline
        (d) {Small Sample Sizes:} $N_1(20,10,25)$, $N_5(20,15,35,25)$ & High heterogeneous: $(3,1,6)$, $(3,1,6,4)$\\
        \hline
    \end{tabular}
    \end{footnotesize}
\end{table}

\begin{table}[!ht]
    \caption{Comparison of the estimated sizes of LRT, Max-D, Min-D tests, and PB-AZ test  for $k=2,\;\underline{\boldsymbol{\mu}}=(1,1,1)$, and $\alpha=0.05$}   
    \label{Table 1}
    \vspace{2mm}
    \begin{footnotesize}
    \centering
        \begin{tabular}{|c|c|c|}
            \hline
             &$N_1(5,5,5)$ &$N_2(20,10,25)$ \\
            \hline
            {$\boldsymbol{\underline{\sigma}^2}$} & {\;\; Max-D \;\;\; LRT \;\;\; Min-D \;\; PB-AZ } & {\;\; Max-D \;\;\; LRT \;\;\; Min-D \;\; PB-AZ } \\
            \hline
    (0.5,1,1.5) & 0.0485 \;\; 0.0498 \;\; 0.0504 \;\; 0.0424 & 0.0534 \;\; 0.0541 \;\; 0.0530 \;\; 0.0496 \\
        (1,2,5) & 0.0511 \;\; 0.0514 \;\; 0.0499 \;\; 0.0464 & 0.0519 \;\; 0.0527 \;\; 0.0534 \;\; 0.0516 \\
       (4,8,10) & 0.0483 \;\; 0.0498 \;\; 0.0509 \;\; 0.0448 & 0.0531 \;\; 0.0535 \;\; 0.0532 \;\; 0.0496 \\
       (15,3,2) & 0.0511 \;\; 0.0500 \;\; 0.0616 \;\; 0.0560 & 0.0522 \;\; 0.0520 \;\; 0.0525 \;\; 0.0472 \\
        (1,1,1) & 0.0506 \;\; 0.0497 \;\; 0.0551 \;\; 0.0432 & 0.0526 \;\; 0.0530 \;\; 0.0513 \;\; 0.0512 \\
    (2.5,4,5.5) & 0.0482 \;\; 0.0503 \;\; 0.0503 \;\; 0.0432 & 0.0523 \;\; 0.0528 \;\; 0.0529 \;\; 0.0504 \\
      (13,10,7) & 0.0511 \;\; 0.0495 \;\; 0.0565 \;\; 0.0460 & 0.0545 \;\; 0.0529 \;\; 0.0503 \;\; 0.0492 \\
       (15,5,3) & 0.0505 \;\; 0.0491 \;\; 0.0606 \;\; 0.0508 & 0.0524 \;\; 0.0523 \;\; 0.0513 \;\; 0.0476 \\
  (0.1,0.2,0.3) & 0.0485 \;\; 0.0498 \;\; 0.0504 \;\; 0.0424 & 0.0534 \;\; 0.0541 \;\; 0.0530 \;\; 0.0496 \\
      (30,35,5) & 0.0525 \;\; 0.0533 \;\; 0.0549 \;\; 0.0520 & 0.0536 \;\; 0.0525 \;\; 0.0500 \;\; 0.0512 \\
       (1,1,20) & 0.0524 \;\; 0.0541 \;\; 0.0486 \;\; 0.0460 & 0.0514 \;\; 0.0510 \;\; 0.0514 \;\; 0.0524 \\
     (50,60,70) & 0.0493 \;\; 0.0502 \;\; 0.0526 \;\; 0.0404 & 0.0527 \;\; 0.0529 \;\; 0.0511 \;\; 0.0488 \\
          \hline
    \end{tabular}
    \end{footnotesize}
\end{table}

\begin{table}[!ht]
    \caption{Comparison of the estimated sizes of LRT, Max-D, Min-D tests, and PB-AZ test  for $k=3,\;\underline{\boldsymbol{\mu}}=(1,1,1,1)$, and $\alpha=0.05$}   
    \label{Table 2}
    \begin{footnotesize}
    \vspace{2mm}
    \centering
        \begin{tabular}{|c|c|c|}
            \hline
             &$N_3(5,8,12,10)$ &$N_4(20,15,35,25)$ \\
            \hline
            {$\boldsymbol{\underline{\sigma}^2}$} & {\;\; Max-D \;\;\; LRT \;\;\; Min-D \;\; PB-AZ } & {\;\; Max-D \;\;\; LRT \;\;\; Min-D \;\; PB-AZ } \\
            \hline
     (.05,.06,.08,.1) & 0.0466 \;\; 0.0522 \;\; 0.0607 \;\; 0.0468 & 0.0511 \;\; 0.0498 \;\; 0.0506 \;\; 0.0512 \\
            (1,4,3,3) & 0.0454 \;\; 0.0496 \;\; 0.0556 \;\; 0.0512 & 0.0504 \;\; 0.0496 \;\; 0.0504 \;\; 0.0520 \\
            (2,2,2,2) & 0.0483 \;\; 0.0513 \;\; 0.0619 \;\; 0.0544 & 0.0507 \;\; 0.0493 \;\; 0.0508 \;\; 0.0468 \\
            (4,1,1,2) & 0.0477 \;\; 0.0532 \;\; 0.0650 \;\; 0.0464 & 0.0487 \;\; 0.0498 \;\; 0.0498 \;\; 0.0492 \\
            (1,6,6,1) & 0.0458 \;\; 0.0485 \;\; 0.0524 \;\; 0.0492 & 0.0500 \;\; 0.0488 \;\; 0.0508 \;\; 0.0516 \\
            (4,6,8,9) & 0.0457 \;\; 0.0515 \;\; 0.0597 \;\; 0.0540 & 0.0517 \;\; 0.0507 \;\; 0.0498 \;\; 0.0504 \\
            (7,9,8,3) & 0.0489 \;\; 0.0511 \;\; 0.0630 \;\; 0.0560 & 0.0502 \;\; 0.0492 \;\; 0.0498 \;\; 0.0540 \\
         (15,11,23,5) & 0.0499 \;\; 0.0532 \;\; 0.0614 \;\; 0.0496 & 0.0481 \;\; 0.0500 \;\; 0.0524 \;\; 0.0520 \\
        (20,30,40,50) & 0.0459 \;\; 0.0519 \;\; 0.0589 \;\; 0.0556 & 0.0516 \;\; 0.0504 \;\; 0.0497 \;\; 0.0496 \\
          (15,12,8,2) & 0.0491 \;\; 0.0513 \;\; 0.0649 \;\; 0.0528 & 0.0495 \;\; 0.0511 \;\; 0.0511 \;\; 0.0520 \\
         (10,10,8,10) & 0.0483 \;\; 0.0517 \;\; 0.0622 \;\; 0.0520 & 0.0509 \;\; 0.0491 \;\; 0.0502 \;\; 0.0512 \\
          (1,1,20,20) & 0.0487 \;\; 0.0504 \;\; 0.0488 \;\; 0.0492 & 0.0509 \;\; 0.0513 \;\; 0.0494 \;\; 0.0524 \\
          \hline
    \end{tabular} 
    \end{footnotesize}
\end{table}

\begin{table}[ht]
    \caption{Comparison of the estimated sizes of LRT, Max-D, Min-D tests, and PB-AZ test  for $k=4,\;\underline{\boldsymbol{\mu}}=(1,1,1,1,1)$, and $\alpha=0.05$}   
    \label{Table 3}
    \begin{footnotesize}
    \vspace{2mm}
    \centering
        \begin{tabular}{|c|c|c|}
            \hline
            &$N_{5}(10,10,10,10,10)$ &$N_{6}(5,5,5,20,15)$ \\
            \hline
            {$\boldsymbol{\underline{\sigma}^2}$} & {\;\; Max-D \;\;\; LRT \;\;\; Min-D \;\; PB-AZ } & {\;\; Max-D \;\;\; LRT \;\;\; Min-D \;\; PB-AZ } \\
            \hline
      (.05,.06,.08,.1,1) & 0.0523 \;\; 0.0505 \;\; 0.0567 \;\; 0.0480 & 0.0490 \;\; 0.0562 \;\; 0.0565 \;\; 0.0416 \\
          (5,8,12,18,20) & 0.0529 \;\; 0.0529 \;\; 0.0541 \;\; 0.0484 & 0.0513 \;\; 0.0558 \;\; 0.0570 \;\; 0.0468 \\
          (18,14,10,8,3) & 0.0512 \;\; 0.0551 \;\; 0.0577 \;\; 0.0476 & 0.0503 \;\; 0.0578 \;\; 0.0627 \;\; 0.0504 \\
        (80,65,55,30,10) & 0.0509 \;\; 0.0548 \;\; 0.0565 \;\; 0.0496 & 0.0498 \;\; 0.0572 \;\; 0.0621 \;\; 0.0496 \\
        (50,50,50,50,50) & 0.0515 \;\; 0.0530 \;\; 0.0557 \;\; 0.0428 & 0.0528 \;\; 0.0584 \;\; 0.0643 \;\; 0.0500 \\
        (10,15,15,15,10) & 0.0500 \;\; 0.0531 \;\; 0.0534 \;\; 0.0456 & 0.0504 \;\; 0.0577 \;\; 0.0601 \;\; 0.0508 \\
        (50,60,70,80,90) & 0.0527 \;\; 0.0534 \;\; 0.0555 \;\; 0.0456 & 0.0506 \;\; 0.0567 \;\; 0.0624 \;\; 0.0496 \\
         (40,40,40,40,4) & 0.0522 \;\; 0.0543 \;\; 0.0546 \;\; 0.0464 & 0.0502 \;\; 0.0571 \;\; 0.0607 \;\; 0.0488 \\
          (5,6,45,50,50) & 0.0543 \;\; 0.0524 \;\; 0.0536 \;\; 0.0484 & 0.0515 \;\; 0.0551 \;\; 0.0500 \;\; 0.0428 \\
            (2,2,2,2,40) & 0.0534 \;\; 0.0528 \;\; 0.0579 \;\; 0.0468 & 0.0515 \;\; 0.0587 \;\; 0.0594 \;\; 0.0472 \\
        (10,20,30,20,10) & 0.0518 \;\; 0.0527 \;\; 0.0526 \;\; 0.0492 & 0.0525 \;\; 0.0580 \;\; 0.0553 \;\; 0.0500 \\
           (10,10,1,1,1) & 0.0500 \;\; 0.0537 \;\; 0.0563 \;\; 0.0496 & 0.0493 \;\; 0.0563 \;\; 0.0636 \;\; 0.0500 \\
          \hline
    \end{tabular}
    \end{footnotesize}
    
\end{table}

\begin{figure}[!ht]
    \centering
    \includegraphics[width=1\textwidth]{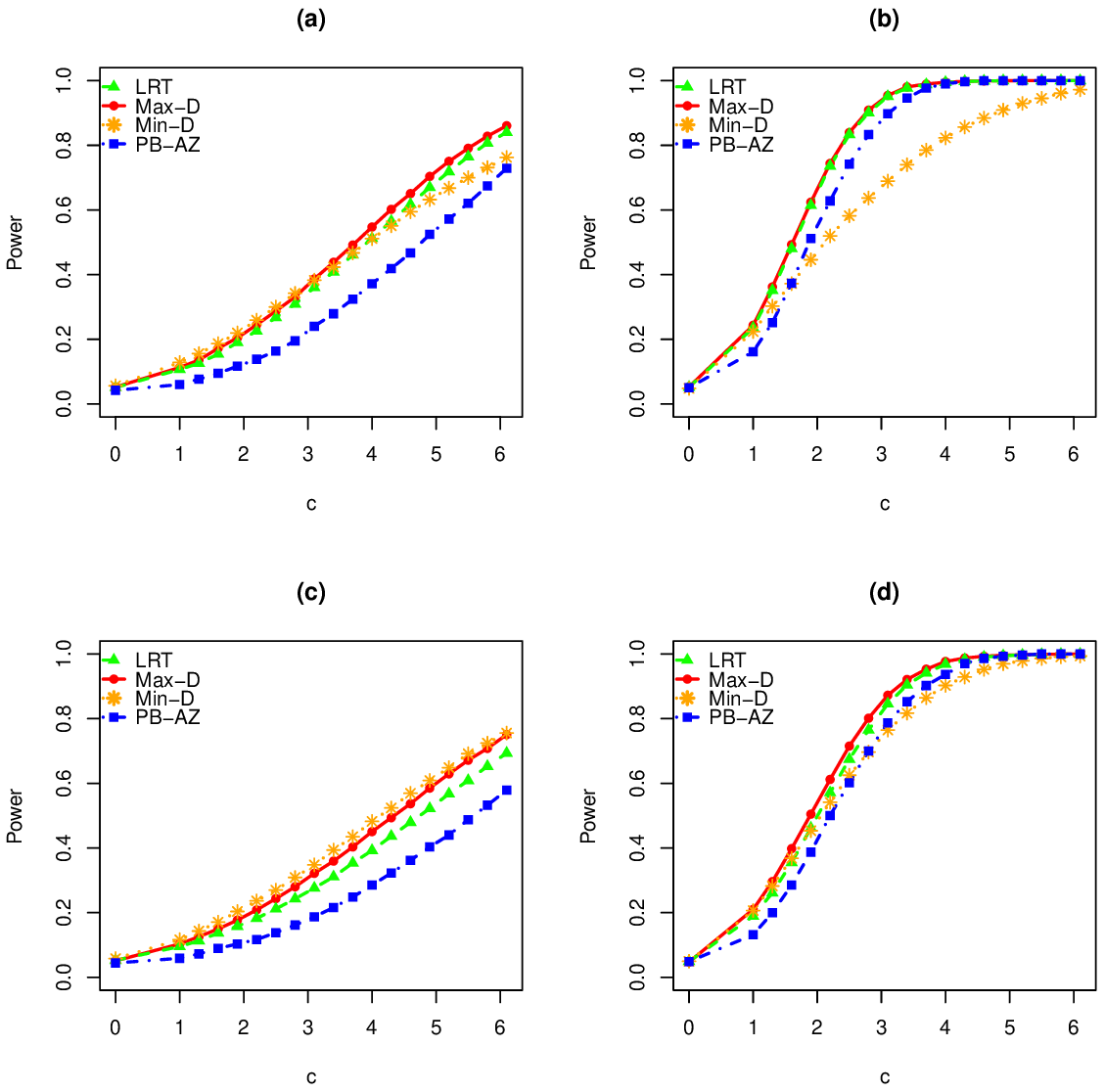}
    \caption{Power curves of the LRT, Max-D, Min-D, and PB-AZ tests when $k = 2$ with $(\mu_0, \mu_1, \mu_2) = c(1, 1.3, 1.6)$: (a) For moderate heterogeneous variances and small sample sizes, (b) For moderate heterogeneous variances and moderate sample sizes, (c) For high heterogeneous variances and small sample sizes, (d) For high heterogeneous variances and moderate sample sizes}
    \label{fig:1}
\end{figure}

.
\begin{figure}[!ht]
    \centering
    \includegraphics[width=1\textwidth]{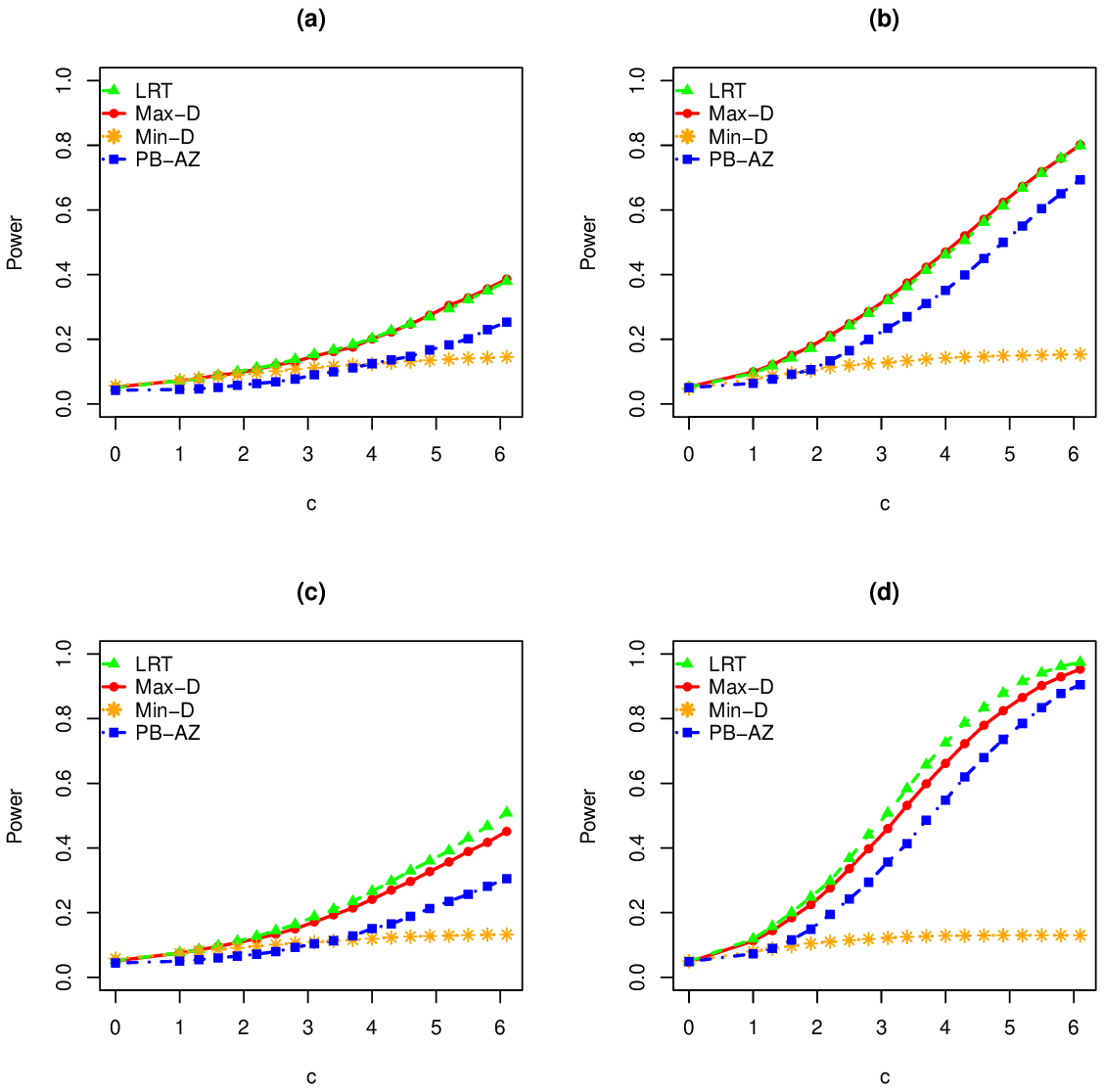}
    \caption{Power curves of the LRT, Max-D, Min-D, and PB-AZ tests when $k = 2$ with $(\mu_0, \mu_1, \mu_2) = c(1, 1.3, 1)$: (a) For moderate heterogeneous variances and small sample sizes, (b) For moderate heterogeneous variances and moderate sample sizes, (c) For high heterogeneous variances and small sample sizes, (d) For high heterogeneous variances and moderate sample sizes}
    \label{fig:2}
\end{figure}

\begin{figure}[!ht]
    \centering
    \includegraphics[width=1\textwidth]{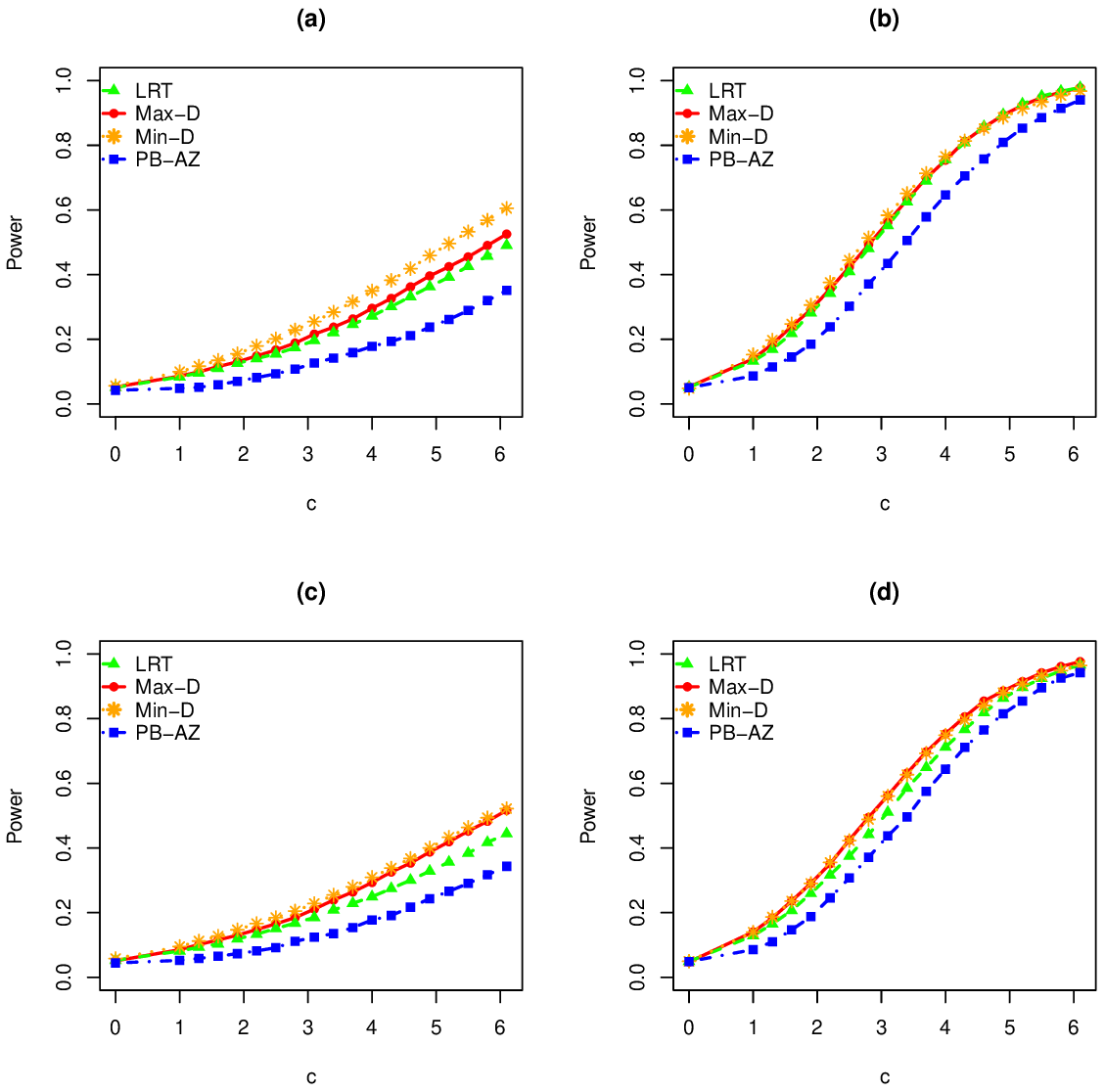}
    \caption{Power curves of the LRT, Max-D, Min-D, and PB-AZ tests when $k = 2$ with $(\mu_0, \mu_1, \mu_2) = c(1, 1.3, 1.3)$: (a) For moderate heterogeneous variances and small sample sizes, (b) For moderate heterogeneous variances and moderate sample sizes, (c) For high heterogeneous variances and small sample sizes, (d) For high heterogeneous variances and moderate sample sizes}
    \label{fig:3}
\end{figure}

\begin{figure}[!ht]
    \centering
    \includegraphics[width=0.94\textwidth]{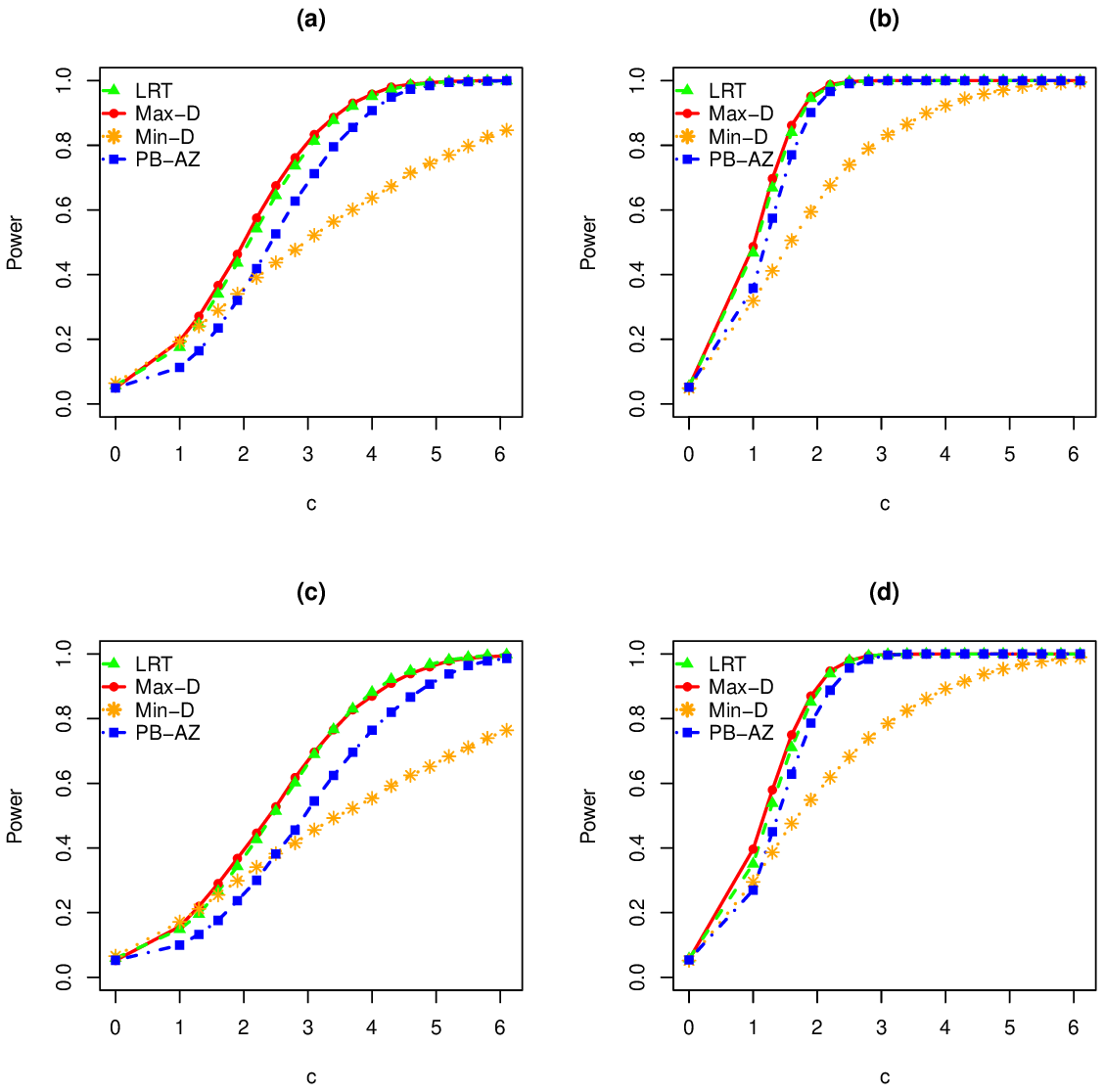}
    \caption{Power curves of the LRT, Max-D, Min-D, and PB-AZ tests when $k = 3$ with $(\mu_0, \mu_1, \mu_2, \mu_3) = c(1, 1.6, 1.9, 1.3)$: (a) For moderate heterogeneous variances and small sample sizes, (b) For moderate heterogeneous variances and moderate sample sizes, (c) For high heterogeneous variances and small sample sizes, (d) For high heterogeneous variances and moderate sample sizes}
    \label{fig:4}
\end{figure}

\begin{figure}[!ht]
    \centering
    \includegraphics[width=0.94\textwidth]{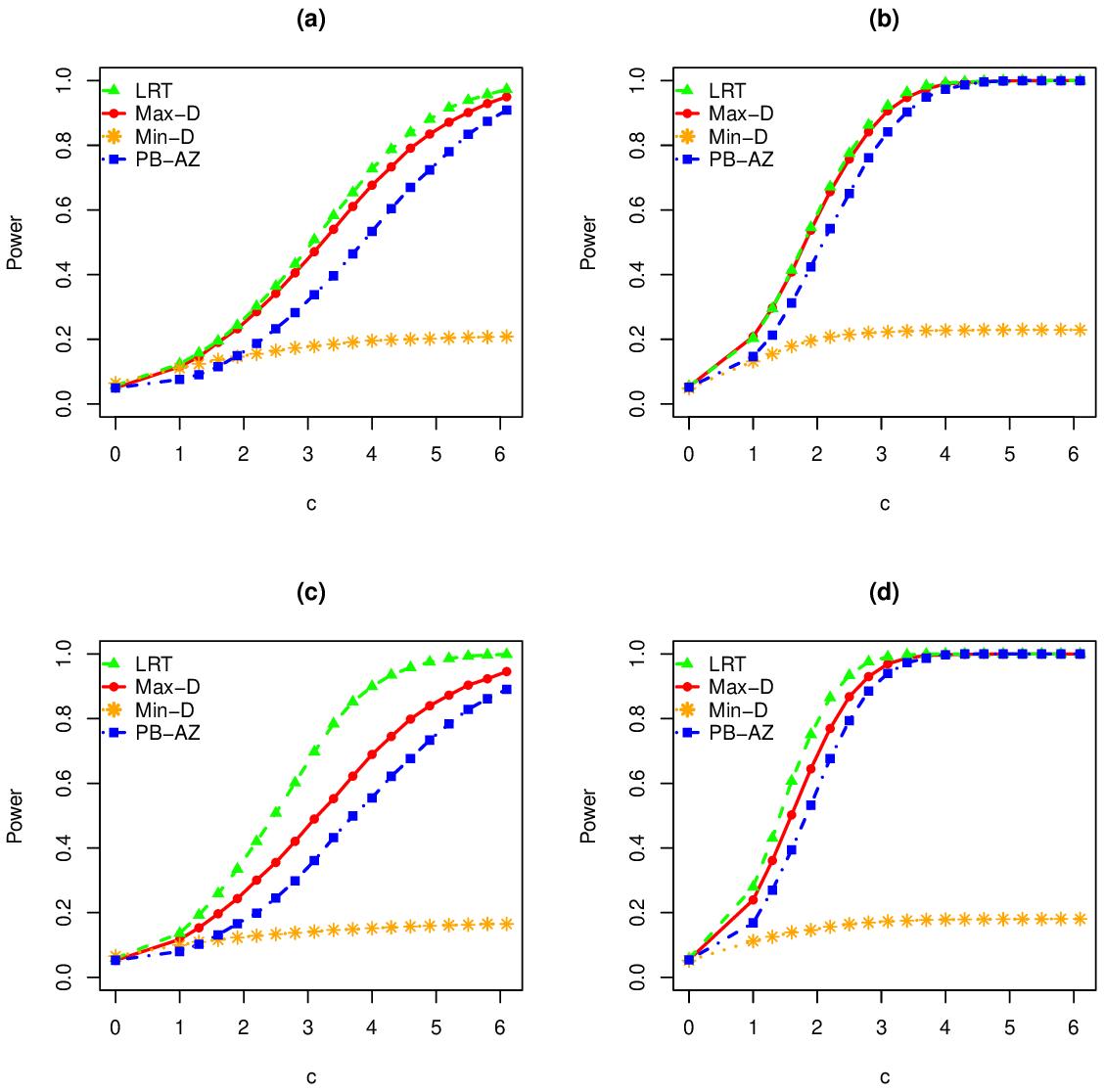}
    \caption{Power curves of the LRT, Max-D, Min-D, and PB-AZ tests when $k = 3$ with $(\mu_0, \mu_1, \mu_2, \mu_3) = c(1, 1.6, 1.3, 1)$: (a) For moderate heterogeneous variances and small sample sizes, (b) For moderate heterogeneous variances and moderate sample sizes, (c) For high heterogeneous variances and small sample sizes, (d) For high heterogeneous variances and moderate sample sizes}
    \label{fig:5}
\end{figure}

\begin{figure}[!ht]
    \centering
    \includegraphics[width=0.94\textwidth]{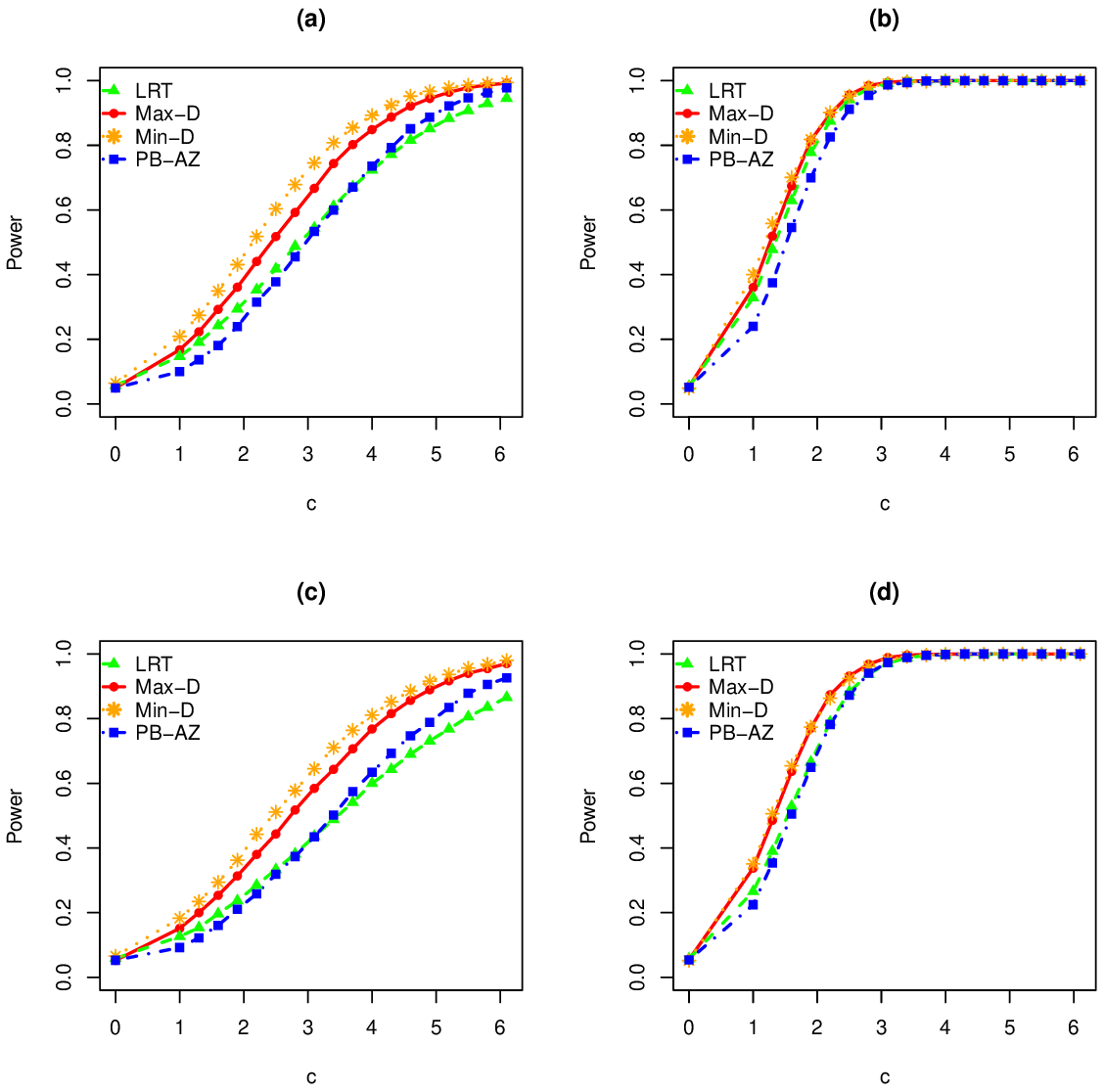}
    \caption{Power curves of the LRT, Max-D, Min-D, and PB-AZ tests when $k = 3$ with $(\mu_0, \mu_1, \mu_2, \mu_3) = c(1, 1.6, 1.6, 1.6)$: (a) For moderate heterogeneous variances and small sample sizes, (b) For moderate heterogeneous variances and moderate sample sizes, (c) For high heterogeneous variances and small sample sizes, (d) For high heterogeneous variances and moderate sample sizes}
    \label{fig:6}
\end{figure}

\begin{enumerate}
\item [(i)] From Tables \ref{Table 1} - \ref{Table 3}, it can be observed that the empirical sizes of PB-based LRT, Max-D, and Min-D are very close to the nominal size $0.05$ in almost all cases. Only Min-D is slightly liberal for very small samples. On the other hand, the PB-AZ test is conservative for small samples and preserves nominal size for moderate sample size cases. 

\item[(ii)] In Fig. \ref{fig:1}, we have taken the choice of mean effects to be strictly monotone. It is observed that Max-D has the best performance throughout, except when the sample sizes are small and variance heterogeneity is high. In these cases, the power values of LRT are very close to those of Max-D. In the latter case, Min-D has higher powers for initial values of $c$, but as $c$ increases, Max-D beats Min-D. In all cases of parameter configurations, LRT and Max-D have superior power than the PB-AZ test. For moderate sample sizes, Min-D test has higher powers than the PB-AZ test for small values of $c$, but as $c$ increases, the PB-AZ test beats Min-D.

\item[(iii)] Fig. \ref{fig:2} deals with the case when one of the treatment means coincides with the control mean. Here, LRT performs the best, with Max-D being very close. Note that Min-D does not have good power values in this case. This is due to the fact that it may take values very close to zero. As $c$ increases, the powers of LRT and Max-D approach unity. When variances are moderately heterogeneous, Max-D and LRT have similar power values, whereas, for highly heterogeneous variances, LRT performs better than Max-D. It is further noted for this case that when $k$ increases, then the performance of LRT improves over Max-D and Min-D, but Max-D remains quite superior to Min-D (see Fig. \ref{fig:5}). For each of the parameter configurations (Fig. \ref{fig:2} and Fig. \ref{fig:5}) Max-D and LRT exhibit greater powers than the PB-AZ test.
    
\item[(iv)] In Fig. \ref{fig:3}, all treatment means are taken to be equal but more than the control. In this situation, Min-D performs better than Max-D and LRT in all cases, except when sample sizes are moderate and variances are highly heterogeneous. In the last case, Min-D and Max-D have almost identical power curves, and LRT has lower powers for some values of $c$. Here our proposed tests beat the PB-AZ in each of the parameter configurations. The power behavior of our proposed tests remains similar in this case when the value of $k$ increases (see Fig. \ref{fig:6}). For higher $k$, PB-AZ shows higher power than Min-D as $c$ increases except in the moderate sample case.

\item[(v)] Fig. \ref{fig:4} considers the case when all treatment means are distinct and strictly larger than the control. Here, Max-D and LRT have almost identical power curves and are quite better than Min-D and PB-AZ in all cases. In all cases, the Min-D test has higher powers than the PB-AZ test for initial values of $c$, but as $c$ increases, the PB-AZ test beats Min-D.

\item[(vi)] It appears that the power of all tests increases as the value of $c$ increases. Power is dependent on ${\delta\mu_i}'s$, where $\delta\mu_i = \mu_{i} - \mu_0$, $i=1,2,\ldots,k$. The power of the tests varies for various combinations of $(\delta\mu_1,\ldots, \delta\mu_{k})$ when sample sizes and group variances are kept fixed. An increase in the variance of any group leads to a drop in power, assuming all other parameters remain constant. Powers also increase as the sample sizes increase.

\item[(vii)] It may be concluded that all three tests achieve the intended size, except that Min-D is liberal and PB-AZ is conservative for very small samples. Furthermore, Max-D and LRT demonstrate higher power than PB-AZ across all parameter configurations, and they outperform Min-D in every scenario except when all treatment means coincide but exceed the mean of the control group. Implementation of Max-D and Min-D is quite easy for any number of treatments, whereas LRT is computationally intensive.

\end{enumerate}

 \begin{remark}
    The packages TreeMaxD, TreeMinD, and TreeLRT have been developed to conduct Max-D, Min-D, and LRT tests, respectively, for any number of groups. The functions that relate to this are TreeMaxD, TreeMinD, and TreeLRT. These files are uploaded to the first author's GitHub account, `SubhaHalder-spec'. To perform the tests, one needs to follow the instructions given below.
\end{remark}

\begin{verbatim}
install.packages("devtools")
library(devtools)
install_github("SubhaHalder-spec/TreeMaxD", force = TRUE)
library(TreeMaxD)
install_github("SubhaHalder-spec/TreeMinD", force = TRUE)
library(TreeMinD)
install_github("SubhaHalder-spec/TreeLRT", force = TRUE)
library(TreeLRT)
TreeMaxD(list(data), alpha)
TreeMinD(list(data), alpha)
TreeLRT(list(data), alpha),
\end{verbatim}
where data $= (X_{0,n_0}, X_{1,n_1},\ldots, X_{k,n_k}),$ and alpha $(\alpha)$ is the significance level.

\section{Robustness of Tests}\label{Robustness}
In this section, we analyze the robustness of the three tests suggested in Sections \ref{LRT} and \ref{Bootstrap Max-D Min-D}. We examine the impact of sampling from non-normal distributions on the size and power values of the LRT, Max-D, and Min-D tests. Five non-normal distributions are considered - skew-normal distribution (with location parameter 1, scale parameter 1, and shape parameter 1), t distribution (with degrees of freedom 3), Laplace distribution (with location parameter 0 and scale parameter 1), a mixture of N(1,1) and N(5,16) with mixing parameter $(0.5,0.5)$ and exponential distribution (with rate parameter 2). For computing empirical size and power values, the configurations for all non-normal distributions are chosen to be similar to those for the normal considered in Section \ref{Size and Power}. The algorithms given in earlier sections are used with some modifications. In Step 2 of Algorithms \ref{Algorithm 2}, we generate samples from given non-normal distributions and standardize these as $W_{ij}$. Now take $X_{ij}=\mu_i + \sigma_i W_{ij}$. These $X_{ij}$'s are then used to obtain the size and power values of LRT, Max-D, and Min-D tests by following the remaining steps of Algorithms \ref{Algorithm 2} and \ref{Algorithm 3}.

As shown in Tables \ref{Table 4} - \ref{Table 6}, size calculations for all the tests are based on extremely unbalanced and large group variances $(\underline{\boldsymbol{\sigma}}^2=(8^2,10^2,12^2), (11^2,8^2,5^2))$ and various sample sizes, including very small, medium, unbalanced, large, and equal sample sizes. The power plots of the tests are presented for the distributions (normal, skew-normal, t, and Laplace) for which type I errors are controlled. The choice of mean vector is taken as $(\underline{\boldsymbol\mu}=c(1, 1.3, 1.6))$ with two specific cases of other parameters as (i) Moderate heterogeneous variances and very small sample size $(\underline{\boldsymbol{\sigma}}^2=(2,3,4) \;\text{and}\; N_1(5,5,5))$, (ii) Moderate heterogeneous variances and moderate sample sizes $(\underline{\boldsymbol{\sigma}}^2=(2,3,4)\;\text{and}\; N_2(20, 10, 25))$.  

We use the symbols N, Sk-N, T, L, Mix-N, and Ex to represent normal, skew-normal, students’ t, Laplace or double exponential, mixture-normal, and exponential distributions, respectively, in Tables \ref{Table 4} - \ref{Table 6} and Fig. \ref{fig:7} - \ref{fig:9}. The following conclusions can be drawn from these tables and graphs regarding the size and power performance of LRT, Max-D, and Min-D for non-normal distributions.

\begin{table}[!ht]
    \caption{Estimated sizes of LRT for different distributions}
    \label{Table 4}
    \begin{tiny}
    \vspace{3mm}
    \centering \setlength\tabcolsep{0.001pt}
        \begin{tabular}{c c  | c  }
        
     & \multicolumn{1}{c|}{$\boldsymbol{\underline{\sigma}}=(8,10,12)$} & \multicolumn{1}{c}{$\boldsymbol{\underline{\sigma}}=(11,8,5)$} \\
    \hline
    {$(n_1,n_2,n_3)$}&{\hspace{0.2cm}N\hspace{0.6cm}Sk-N\hspace{0.6cm}T\hspace{0.8cm}L\hspace{0.6cm}Mix-N\hspace{0.6cm}Ex \hspace{0.2cm}}\hspace{0.1cm}&{\hspace{0.6cm}N\hspace{0.6cm}Sk-N\hspace{0.6cm}T\hspace{0.8cm}L\hspace{0.6cm}Mix-N\hspace{0.6cm}Ex \hspace{0.2cm}}\hspace{0.1cm} \\  \hline
     (5,5,5) & 0.0499 \;\; 0.0478 \;\; 0.0394 \;\; 0.0364 \;\; 0.0165 \;\; 0.0275 & \;\; 0.0503 \;\; 0.0514 \;\; 0.0413  \;\; 0.0373 \;\; 0.1251 \;\; 0.0983 \\
  (10,10,15) & 0.0513 \;\; 0.0493 \;\; 0.0441 \;\; 0.0489 \;\; 0.0403 \;\; 0.0394 & \;\; 0.0540 \;\; 0.0548 \;\; 0.0444  \;\; 0.0443 \;\; 0.1122 \;\; 0.1099 \\
   (5,15,25) & 0.0555 \;\; 0.0528 \;\; 0.0458 \;\; 0.0520 \;\; 0.1343 \;\; 0.0977 & \;\; 0.0532 \;\; 0.0533 \;\; 0.0450  \;\; 0.0467 \;\; 0.1977 \;\; 0.1522 \\
  (30,15,10) & 0.0471 \;\; 0.0510 \;\; 0.0434 \;\; 0.0454 \;\; 0.0177 \;\; 0.0207 & \;\; 0.0469 \;\; 0.0529 \;\; 0.0442  \;\; 0.0507 \;\; 0.0458 \;\; 0.0499 \\
  (10,50,40) & 0.0497 \;\; 0.0567 \;\; 0.0466 \;\; 0.0486 \;\; 0.1017 \;\; 0.0942 & \;\; 0.0519 \;\; 0.0577 \;\; 0.0474  \;\; 0.0485 \;\; 0.1179 \;\; 0.1157 \\
    (5,30,5) & 0.0538 \;\; 0.0558 \;\; 0.0459 \;\; 0.0413 \;\; 0.1697 \;\; 0.1181 & \;\; 0.0539 \;\; 0.0598 \;\; 0.0467  \;\; 0.0409 \;\; 0.1885 \;\; 0.1309 \\
     (8,7,6) & 0.0513 \;\; 0.0506 \;\; 0.0399 \;\; 0.0438 \;\; 0.0210 \;\; 0.0298 & \;\; 0.0544 \;\; 0.0514 \;\; 0.0403  \;\; 0.0432 \;\; 0.0855 \;\; 0.0845 \\
   (30,45,8) & 0.0492 \;\; 0.0499 \;\; 0.0437 \;\; 0.0458 \;\; 0.0366 \;\; 0.0337 & \;\; 0.0527 \;\; 0.0535 \;\; 0.0447  \;\; 0.0463 \;\; 0.0740 \;\; 0.0714 \\
  (40,55,70) & 0.0515 \;\; 0.0534 \;\; 0.0500 \;\; 0.0482 \;\; 0.0527 \;\; 0.0545 & \;\; 0.0505 \;\; 0.0546 \;\; 0.0510  \;\; 0.0457 \;\; 0.0703 \;\; 0.0807 \\
  (80,65,55) & 0.0498 \;\; 0.0495 \;\; 0.0470 \;\; 0.0518 \;\; 0.0416 \;\; 0.0390 & \;\; 0.0513 \;\; 0.0504 \;\; 0.0485  \;\; 0.0494 \;\; 0.0597 \;\; 0.0639 \\
  (50,50,50) & 0.0488 \;\; 0.0527 \;\; 0.0478 \;\; 0.0502 \;\; 0.0451 \;\; 0.0446 & \;\; 0.0507 \;\; 0.0521 \;\; 0.0462  \;\; 0.0490 \;\; 0.0646 \;\; 0.0713 \\
  (10,10,80) & 0.0508 \;\; 0.0558 \;\; 0.0446 \;\; 0.0449 \;\; 0.1115 \;\; 0.1005 & \;\; 0.0516 \;\; 0.0570 \;\; 0.0459  \;\; 0.0437 \;\; 0.1499 \;\; 0.1485 \\
 (80,95,100) & 0.0472 \;\; 0.0549 \;\; 0.0487 \;\; 0.0490 \;\; 0.0500 \;\; 0.0444 & \;\; 0.0477 \;\; 0.0549 \;\; 0.0460  \;\; 0.0494 \;\; 0.0641 \;\; 0.0615 \\
(100,15,100) & 0.0459 \;\; 0.0469 \;\; 0.0483 \;\; 0.0497 \;\; 0.0297 \;\; 0.0286 & \;\; 0.0467 \;\; 0.0506 \;\; 0.0455  \;\; 0.0490 \;\; 0.0541 \;\; 0.0609 \\
  (10,11,12) & 0.0522 \;\; 0.0531 \;\; 0.0461 \;\; 0.0470 \;\; 0.0381 \;\; 0.0389 & \;\; 0.0557 \;\; 0.0526 \;\; 0.0444  \;\; 0.0472 \;\; 0.0991 \;\; 0.1031 \\
  (25,25,25) & 0.0494 \;\; 0.0528 \;\; 0.0456 \;\; 0.0485 \;\; 0.0414 \;\; 0.0412 & \;\; 0.0509 \;\; 0.0539 \;\; 0.0468  \;\; 0.0480 \;\; 0.0727 \;\; 0.0757 \\

            \hline
        \end{tabular}
    \end{tiny}
\end{table}

\begin{table}[!ht]
    \caption{Estimated sizes of Max-D for different distributions}
    \label{Table 5}
    \begin{tiny}
    \vspace{3mm}
    \centering \setlength\tabcolsep{0.001pt}
        \begin{tabular}{c c  | c  }
        
     & \multicolumn{1}{c|}{$\boldsymbol{\underline{\sigma}}=(8,10,12)$} & \multicolumn{1}{c}{$\boldsymbol{\underline{\sigma}}=(11,8,5)$} \\
    \hline
    {$(n_1,n_2,n_3)$}&{\hspace{0.2cm}N\hspace{0.6cm}Sk-N\hspace{0.6cm}T\hspace{0.8cm}L\hspace{0.6cm}Mix-N\hspace{0.6cm}Ex \hspace{0.2cm}}\hspace{0.1cm}&{\hspace{0.6cm}N\hspace{0.6cm}Sk-N\hspace{0.6cm}T\hspace{0.8cm}L\hspace{0.6cm}Mix-N\hspace{0.6cm}Ex \hspace{0.2cm}}\hspace{0.1cm} \\  \hline
     (5,5,5) & 0.0485 \;\; 0.0469 \;\; 0.0388 \;\; 0.0388 \;\; 0.0890 \;\; 0.0198 & \;\; 0.0507 \;\; 0.0498 \;\; 0.0400  \;\; 0.0419 \;\; 0.1106 \;\; 0.0870 \\
  (10,10,15) & 0.0482 \;\; 0.0505 \;\; 0.0430 \;\; 0.0475 \;\; 0.0298 \;\; 0.0267 & \;\; 0.0500 \;\; 0.0541 \;\; 0.0447  \;\; 0.0482 \;\; 0.0911 \;\; 0.0916 \\
   (5,15,25) & 0.0543 \;\; 0.0522 \;\; 0.0441 \;\; 0.0494 \;\; 0.1179 \;\; 0.0836 & \;\; 0.0529 \;\; 0.0528 \;\; 0.0448  \;\; 0.0451 \;\; 0.1928 \;\; 0.1507 \\
  (30,15,10) & 0.0495 \;\; 0.0501 \;\; 0.0458 \;\; 0.0480 \;\; 0.0159 \;\; 0.0167 & \;\; 0.0497 \;\; 0.0536 \;\; 0.0459  \;\; 0.0517 \;\; 0.0401 \;\; 0.0421 \\
  (10,50,40) & 0.0485 \;\; 0.0540 \;\; 0.0495 \;\; 0.0492 \;\; 0.0982 \;\; 0.0885 & \;\; 0.0461 \;\; 0.0546 \;\; 0.0480  \;\; 0.0474 \;\; 0.1276 \;\; 0.1226 \\
    (5,30,5) & 0.0516 \;\; 0.0538 \;\; 0.0425 \;\; 0.0418 \;\; 0.0974 \;\; 0.0719 & \;\; 0.0519 \;\; 0.0572 \;\; 0.0426  \;\; 0.0447 \;\; 0.1725 \;\; 0.1299 \\
     (8,7,6) & 0.0533 \;\; 0.0509 \;\; 0.0420 \;\; 0.0445 \;\; 0.0128 \;\; 0.0194 & \;\; 0.0520 \;\; 0.0538 \;\; 0.0417  \;\; 0.0451 \;\; 0.0788 \;\; 0.0760 \\
   (30,45,8) & 0.0512 \;\; 0.0507 \;\; 0.0446 \;\; 0.0464 \;\; 0.0272 \;\; 0.0234 & \;\; 0.0512 \;\; 0.0546 \;\; 0.0472  \;\; 0.0504 \;\; 0.0543 \;\; 0.0580 \\
  (40,55,70) & 0.0516 \;\; 0.0516 \;\; 0.0521 \;\; 0.0478 \;\; 0.0497 \;\; 0.0472 & \;\; 0.0529 \;\; 0.0539 \;\; 0.0534  \;\; 0.0499 \;\; 0.0704 \;\; 0.0800 \\
  (80,65,55) & 0.0514 \;\; 0.0465 \;\; 0.0480 \;\; 0.0514 \;\; 0.0399 \;\; 0.0352 & \;\; 0.0488 \;\; 0.0488 \;\; 0.0514  \;\; 0.0501 \;\; 0.0607 \;\; 0.0617 \\
  (50,50,50) & 0.0490 \;\; 0.0504 \;\; 0.0473 \;\; 0.0505 \;\; 0.0406 \;\; 0.0362 & \;\; 0.0495 \;\; 0.0547 \;\; 0.0467  \;\; 0.0518 \;\; 0.0650 \;\; 0.0644 \\
  (10,10,80) & 0.0480 \;\; 0.0527 \;\; 0.0444 \;\; 0.0452 \;\; 0.0807 \;\; 0.0701 & \;\; 0.0476 \;\; 0.0550 \;\; 0.0454  \;\; 0.0466 \;\; 0.1149 \;\; 0.1117 \\
 (80,95,100) & 0.0466 \;\; 0.0553 \;\; 0.0499 \;\; 0.0471 \;\; 0.0493 \;\; 0.0407 & \;\; 0.0483 \;\; 0.0542 \;\; 0.0491  \;\; 0.0479 \;\; 0.0672 \;\; 0.0600 \\
(100,15,100) & 0.0472 \;\; 0.0484 \;\; 0.0456 \;\; 0.0492 \;\; 0.0266 \;\; 0.0251 & \;\; 0.0498 \;\; 0.0540 \;\; 0.0446  \;\; 0.0524 \;\; 0.0442 \;\; 0.0460 \\
  (10,11,12) & 0.0524 \;\; 0.0516 \;\; 0.0426 \;\; 0.0479 \;\; 0.0284 \;\; 0.0297 & \;\; 0.0536 \;\; 0.0517 \;\; 0.0445  \;\; 0.0485 \;\; 0.0914 \;\; 0.0950 \\
  (25,25,25) & 0.0481 \;\; 0.0546 \;\; 0.0480 \;\; 0.0470 \;\; 0.0372 \;\; 0.0326 & \;\; 0.0487 \;\; 0.0566 \;\; 0.0487  \;\; 0.0469 \;\; 0.0734 \;\; 0.0725 \\

            \hline
        \end{tabular}
    \end{tiny}
\end{table}

\begin{table}[!ht]
    \caption{Estimated sizes of Min-D for different distributions}
    \label{Table 6}
    \begin{tiny}
    \vspace{3mm}
    \centering \setlength\tabcolsep{0.001pt}
        \begin{tabular}{c c  | c  }
        
     & \multicolumn{1}{c|}{$\boldsymbol{\underline{\sigma}}=(8,10,12)$} & \multicolumn{1}{c}{$\boldsymbol{\underline{\sigma}}=(11,8,5)$} \\
    \hline
    {$(n_1,n_2,n_3)$}&{\hspace{0.2cm}N\hspace{0.6cm}Sk-N\hspace{0.6cm}T\hspace{0.8cm}L\hspace{0.6cm}Mix-N\hspace{0.6cm}Ex \hspace{0.2cm}}\hspace{0.1cm}&{\hspace{0.6cm}N\hspace{0.6cm}Sk-N\hspace{0.6cm}T\hspace{0.8cm}L\hspace{0.6cm}Mix-N\hspace{0.6cm}Ex \hspace{0.2cm}}\hspace{0.1cm} \\  \hline
     (5,5,5) & 0.0504 \;\; 0.0511 \;\; 0.0541 \;\; 0.0528 \;\; 0.0477 \;\; 0.0609 & \;\; 0.0586 \;\; 0.0620 \;\; 0.0568  \;\; 0.0579 \;\; 0.1621 \;\; 0.1414 \\
  (10,10,15) & 0.0487 \;\; 0.0562 \;\; 0.0516 \;\; 0.0518 \;\; 0.0554 \;\; 0.0690 & \;\; 0.0517 \;\; 0.0616 \;\; 0.0538  \;\; 0.0541 \;\; 0.1023 \;\; 0.1138 \\
   (5,15,25) & 0.0652 \;\; 0.0628 \;\; 0.0542 \;\; 0.0584 \;\; 0.1524 \;\; 0.1307 & \;\; 0.0670 \;\; 0.0638 \;\; 0.0550  \;\; 0.0557 \;\; 0.2000 \;\; 0.1689 \\
  (30,15,10) & 0.0513 \;\; 0.0498 \;\; 0.0508 \;\; 0.0532 \;\; 0.0293 \;\; 0.0414 & \;\; 0.0518 \;\; 0.0544 \;\; 0.0506  \;\; 0.0520 \;\; 0.0560 \;\; 0.0717 \\
  (10,50,40) & 0.0527 \;\; 0.0597 \;\; 0.0534 \;\; 0.0532 \;\; 0.1137 \;\; 0.1119 & \;\; 0.0508 \;\; 0.0593 \;\; 0.0528  \;\; 0.0511 \;\; 0.1360 \;\; 0.1353 \\
    (5,30,5) & 0.0539 \;\; 0.0589 \;\; 0.0549 \;\; 0.0604 \;\; 0.0818 \;\; 0.0953 & \;\; 0.0603 \;\; 0.0706 \;\; 0.0537  \;\; 0.0621 \;\; 0.1957 \;\; 0.1663 \\
     (8,7,6) & 0.0523 \;\; 0.0526 \;\; 0.0517 \;\; 0.0533 \;\; 0.0405 \;\; 0.0542 & \;\; 0.0557 \;\; 0.0615 \;\; 0.0529  \;\; 0.0505 \;\; 0.1105 \;\; 0.1158 \\
   (30,45,8) & 0.0471 \;\; 0.0503 \;\; 0.0539 \;\; 0.0554 \;\; 0.0356 \;\; 0.0479 & \;\; 0.0506 \;\; 0.0528 \;\; 0.0538  \;\; 0.0555 \;\; 0.0660 \;\; 0.0817 \\
  (40,55,70) & 0.0516 \;\; 0.0547 \;\; 0.0585 \;\; 0.0513 \;\; 0.0600 \;\; 0.0635 & \;\; 0.0527 \;\; 0.0531 \;\; 0.0570  \;\; 0.0521 \;\; 0.0784 \;\; 0.0847 \\
  (80,65,55) & 0.0506 \;\; 0.0493 \;\; 0.0529 \;\; 0.0507 \;\; 0.0494 \;\; 0.0488 & \;\; 0.0472 \;\; 0.0511 \;\; 0.0528  \;\; 0.0504 \;\; 0.0640 \;\; 0.0678 \\
  (50,50,50) & 0.0524 \;\; 0.0504 \;\; 0.0482 \;\; 0.0526 \;\; 0.0539 \;\; 0.0583 & \;\; 0.0516 \;\; 0.0552 \;\; 0.0455  \;\; 0.0526 \;\; 0.0709 \;\; 0.0762 \\
  (10,10,80) & 0.0465 \;\; 0.0564 \;\; 0.0556 \;\; 0.0537 \;\; 0.0823 \;\; 0.0918 & \;\; 0.0498 \;\; 0.0584 \;\; 0.0529  \;\; 0.0534 \;\; 0.1106 \;\; 0.1169 \\
 (80,95,100) & 0.0481 \;\; 0.0519 \;\; 0.0531 \;\; 0.0498 \;\; 0.0574 \;\; 0.0555 & \;\; 0.0488 \;\; 0.0548 \;\; 0.0521  \;\; 0.0492 \;\; 0.0700 \;\; 0.0662 \\
(100,15,100) & 0.0488 \;\; 0.0510 \;\; 0.0561 \;\; 0.0513 \;\; 0.0382 \;\; 0.0440 & \;\; 0.0455 \;\; 0.0533 \;\; 0.0537  \;\; 0.0492 \;\; 0.0472 \;\; 0.0586 \\
  (10,11,12) & 0.0515 \;\; 0.0560 \;\; 0.0550 \;\; 0.0524 \;\; 0.0574 \;\; 0.0668 & \;\; 0.0537 \;\; 0.0620 \;\; 0.0561  \;\; 0.0534 \;\; 0.1058 \;\; 0.1184 \\
  (25,25,25) & 0.0526 \;\; 0.0545 \;\; 0.0527 \;\; 0.0484 \;\; 0.0545 \;\; 0.0616 & \;\; 0.0536 \;\; 0.0581 \;\; 0.0514  \;\; 0.0500 \;\; 0.0833 \;\; 0.0891 \\

            \hline
        \end{tabular}
    \end{tiny}
\end{table}

\begin{figure}[!ht]
    \centering
    \includegraphics[width=1\textwidth]{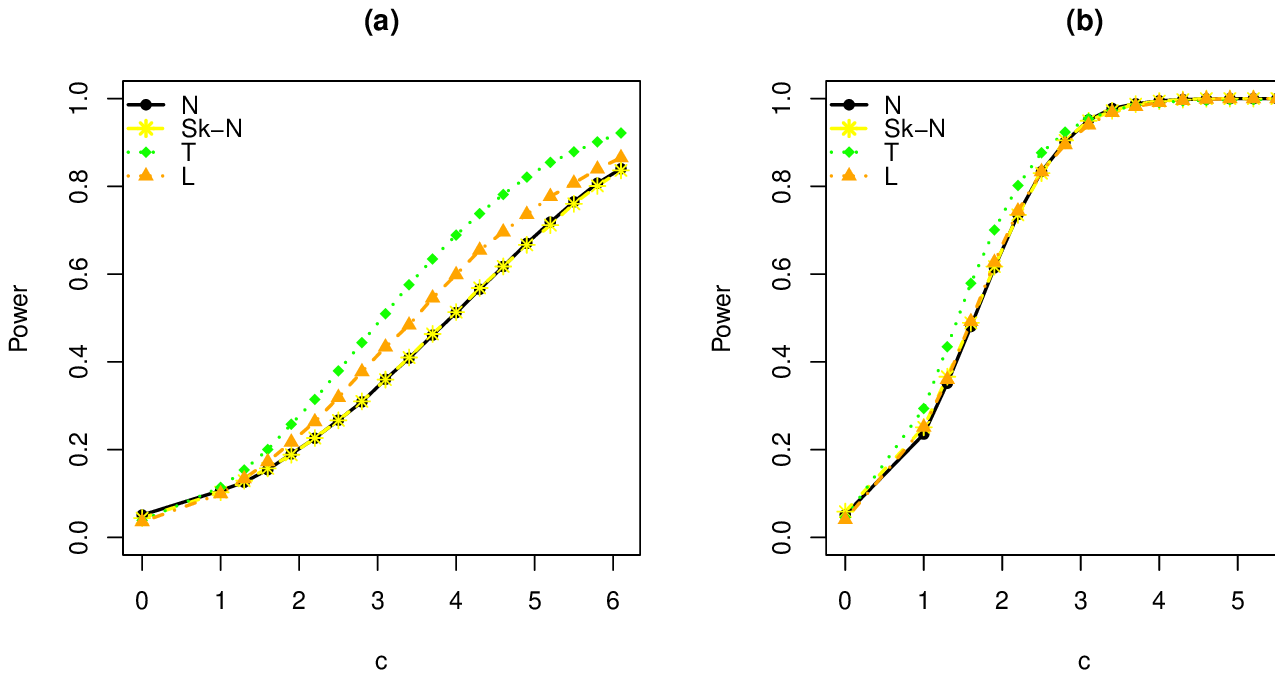}
    \caption{Power plots of the LRT when $k = 2$ with $(\mu_0, \mu_1, \mu_2) = c(1, 1.3, 1.6)$: (a) For high heterogeneous variances and very small sample sizes, (b) For high heterogeneous variances and moderate sample sizes}
    \label{fig:7}
\end{figure}

\begin{figure}[!ht]
    \centering
    \includegraphics[width=1\textwidth]{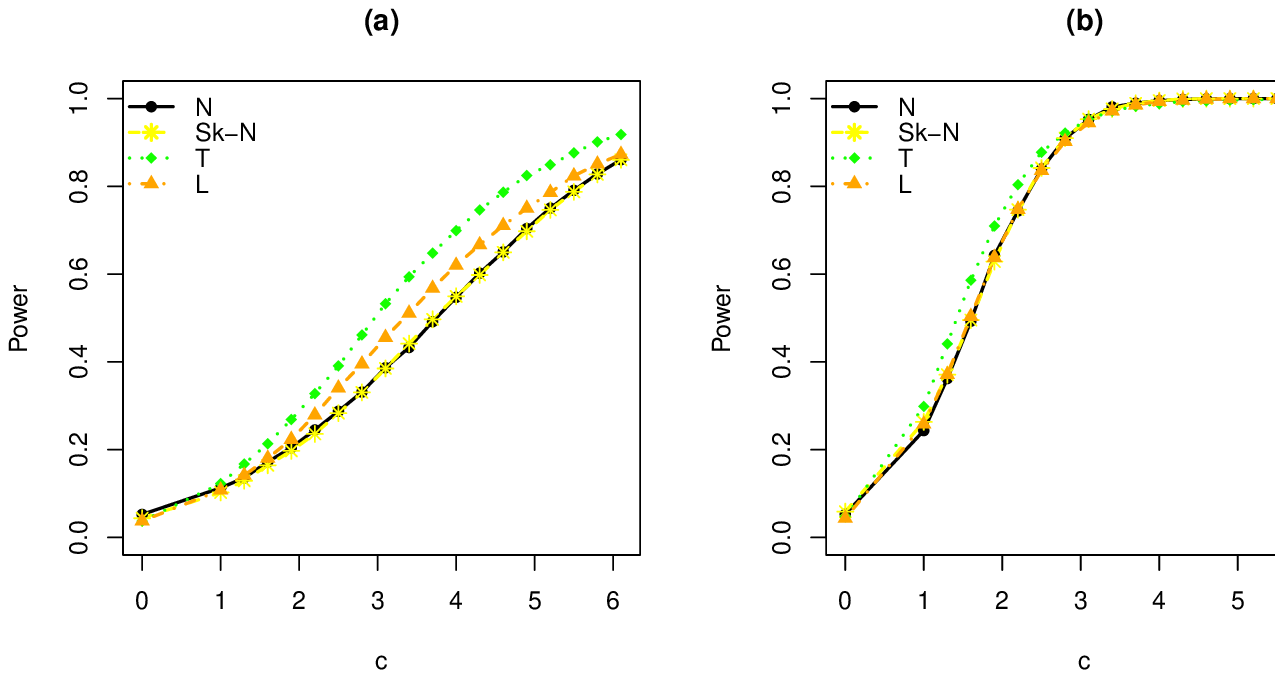}
    \caption{Power plots of the Max-D for different non-normal distributions when $k = 2$ with $(\mu_0, \mu_1, \mu_2) = c(1, 1.3, 1.6)$: (a) For high heterogeneous variances and very small sample sizes, (b) For high heterogeneous variances and moderate sample sizes}
    \label{fig:8}
\end{figure}

\begin{figure}[!ht]
    \centering
    \includegraphics[width=1\textwidth]{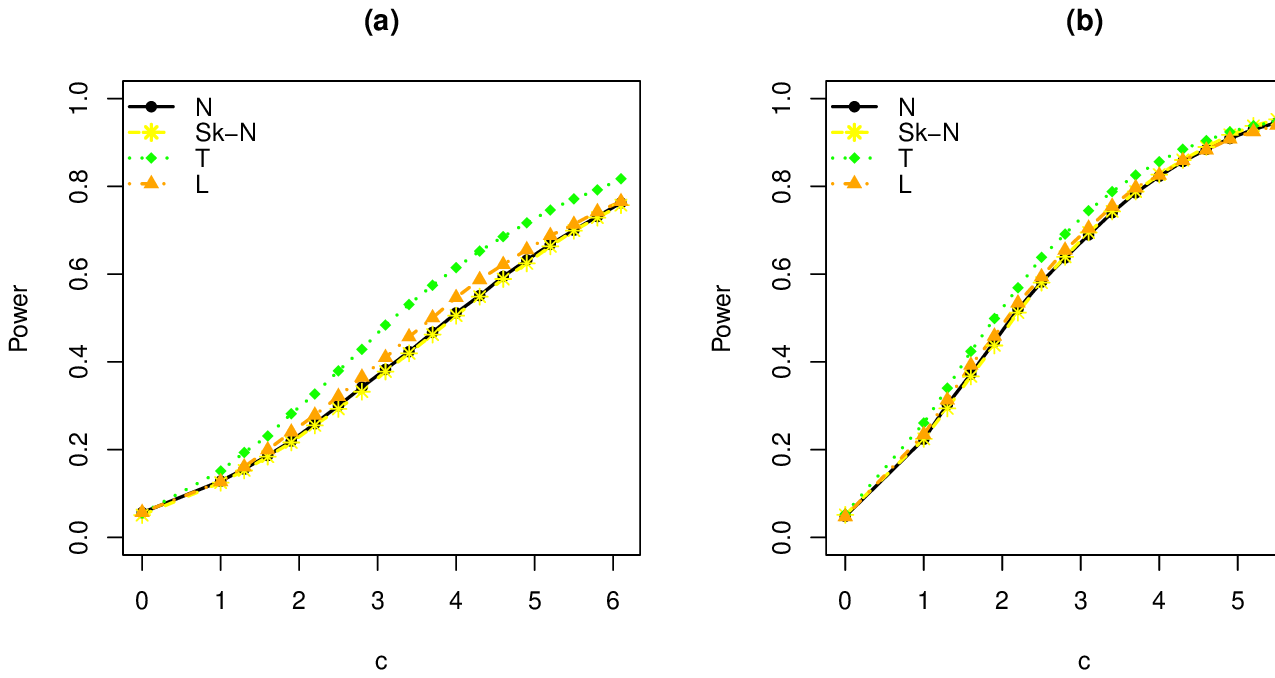}
    \caption{Power plots of the Min-D for different non-normal distributions when $k = 2$ with $(\mu_0, \mu_1, \mu_2) = c(1, 1.3, 1.6)$: (a) For high heterogeneous variances and very small sample sizes, (b) For high heterogeneous variances and moderate sample sizes}
    \label{fig:9}
\end{figure}

\begin{enumerate}
\item [(i)] In Tables \ref{Table 4}-\ref{Table 6}, for highly heterogeneous variances cases is considered. It can be seen that all three tests - the LRT, Max-D, and Min-D achieve a size very close to $0.05$ for skew-normal, Students' t, and Laplace distributions for all sample sizes except when all samples are very small. In this case, LRT and Max-D are slightly conservative for t and Laplace. However, Min-D also achieves the nominal size in this case. For the mixture-normal and exponential distributions, the tests become either very conservative or quite liberal except when all samples are large.

\item[(ii)] No loss of power is observed for any of the tests from Fig. \ref{fig:7} - \ref{fig:9}. In fact, for very small samples, we note small gains in power over the normal distribution. However, it may be due to sampling error for small samples. These trends are also observed for various other configurations of sample sizes, choices of mean effects, and variances. For the sake of brevity, we have not presented these results here.

\end{enumerate}

\section{An Application to an Actual data}\label{Application}
In this section, the proposed tests are implemented on a real dataset from a research study to examine the efficacy of various psychological treatments on the noise sensitivity of individuals who suffer from headaches (one can see the full data from {\small \burl{https://www.kaggle.com/datasets/utkarshx27/headache-sufferers-for-sensitivity-to-noise?resource=download}}). In this experiment, $98$ patients were initially evaluated for their subjective perception of noise loudness, specifically for the levels they found uncomfortable (U) and extremely uncomfortable (EU). Subsequently, the participants underwent relaxation training, during which they were exposed to the noise at the EU level. They were guided on proper breathing methods and instructed on using visual imagery to divert their attention from discomfort. There are four groups: one is Control, and the others are Treatment 1, Treatment 2, and Treatment 3. Here, the control group does not get any psychological treatments between the initial and final measure of noise level. In Treatment 1, patients were exposed again to their initial EU-level tone for the time they previously found tolerable. In Treatment 2, patients were given more time than in Treatment 1 to listen to the tone at their initial EU noise level as before. Treatment 3 patients were given the same time as Treatment 2 to listen to the tone at their initial EU noise level with relaxation techniques. 

Here, we considered the difference in final and initial measures in the noise level rated as Uncomfortable for all patients. The boxplots (Fig. \ref{fig:10}) 
seem to suggest that the mean of the control group is smaller than each treatment mean. Further, the data fits well to the normal distribution, as shown by the Q-Q plot (Fig. \ref{fig:11}). In addition, we performed the Kolmogorov–Smirnov test on this data and obtained p-values as $(0.2210, 0.3146, 0.2430, 0.8497).$  Bartlett's test for the homogeneity of variances yields the p-value 0.0005. So,  the group variances are not homogeneous.

\begin{table}[!ht]
    \caption{Summary Statistics}
    \label{Table 7}
    \vspace{2mm}
    \centering 
    \begin{tabular}{c  c  c  c}
        \hline
        Treatment category & Sample size & Sample mean & Sample variance \\
        \hline
        Control  & 23 & -0.4134783 & 1.416596 \\
        Treatment 1 & 25 & 0.2344000 & 3.422117 \\
        Treatment 2 & 22 & 1.0504545 & 7.297271 \\
        Treatment 3 & 28 & 0.9367857 & 1.935926 \\
        \hline
    \end{tabular}
\end{table}
\begin{table}[!ht]
    \caption{Calculated Size, Power, Test Statistics and Critical Values}
    \label{Table 8}
    \vspace{2mm}
    \centering 
    \begin{tabular}{ c  c  c  c  c }
        \hline
         Tests & (Size, Power)  & Test statistic  & Critical value & Decision   \\
        \hline
         LRT & (0.0556, 0.9655) & 0.0006892 & 0.0440946  & Rejected \\
         Max-D & (0.0563, 0.9678) & 3.7344682 & 2.1667720  & Rejected \\
         Min-D & (0.0526, 0.7667) & 1.4542517 & 0.6552592  & Rejected \\
        \hline
    \end{tabular}
\end{table}

        

Let $\mu_0, \mu_1, \mu_2,$ and $\mu_3$ be the mean of changes in the noise level of the Control, Treatment 1, Treatment 2, and Treatment 3 groups, respectively. The average change in noise level for the control group may be less than that observed in other treatments. Now we want to test $H_0: \mu_0=\mu_1=\mu_2=\mu_3$ against $H_1: \mu_0 \leq \mu_i$ (with at least one strict inequality) for all $i=1,2,3.$ The sample mean and sample variance have been calculated for each group (see Table \ref{Table 7}). Furthermore, the test statistics and critical values have been calculated for the LRT, Max-D, and Min-D tests, as shown in Table \ref{Table 8}. From this table, we may conclude that all three tests reject the null hypothesis. The critical value of the Max-D test for determining simultaneous confidence intervals for $\mu_i - \mu_0$, where $i = 1, 2, 3,$ is $2.1667720$. The resulting confidence intervals are $\mu_1 - \mu_0 \in (-0.3174323 ,\infty)$, $\mu_2 - \mu_0 \in (0.1050967 ,\infty)$, and $\mu_3 - \mu_0 \in (0.5668287 ,\infty)$. Hence, it is clear that $\mu_0 = \mu_1$, $\mu_0 < \mu_2$, and $\mu_0 < \mu_3.$ It implies that Treatment 2 and Treatment 3 are more effective than the control. 
\begin{figure}[!ht]
    \centering
    \includegraphics[width=0.7\textwidth]{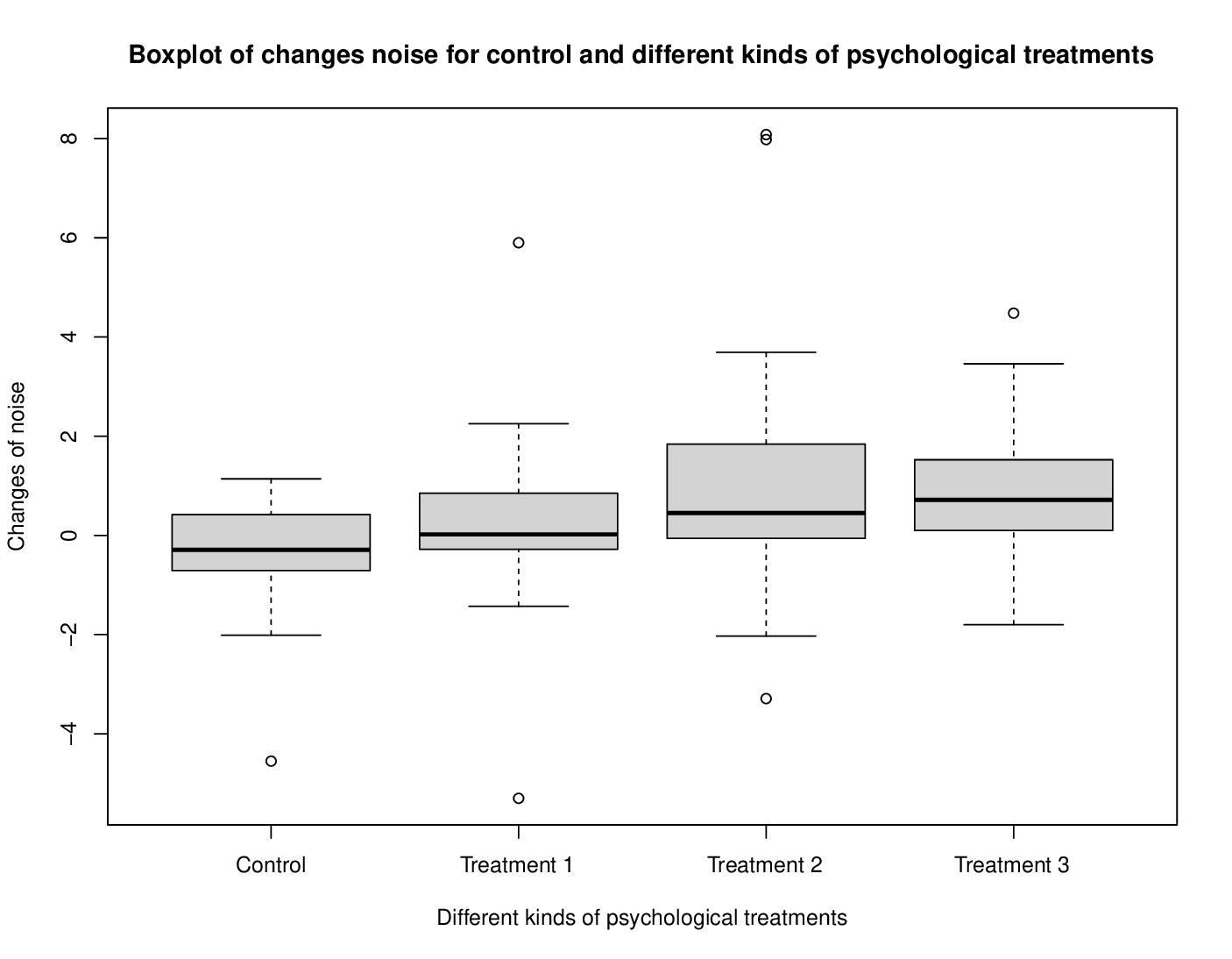}
    \caption{Boxplot of data}
    \label{fig:10}
\end{figure}
\begin{figure}[!ht]
    \centering
    \includegraphics[width=0.9\textwidth]{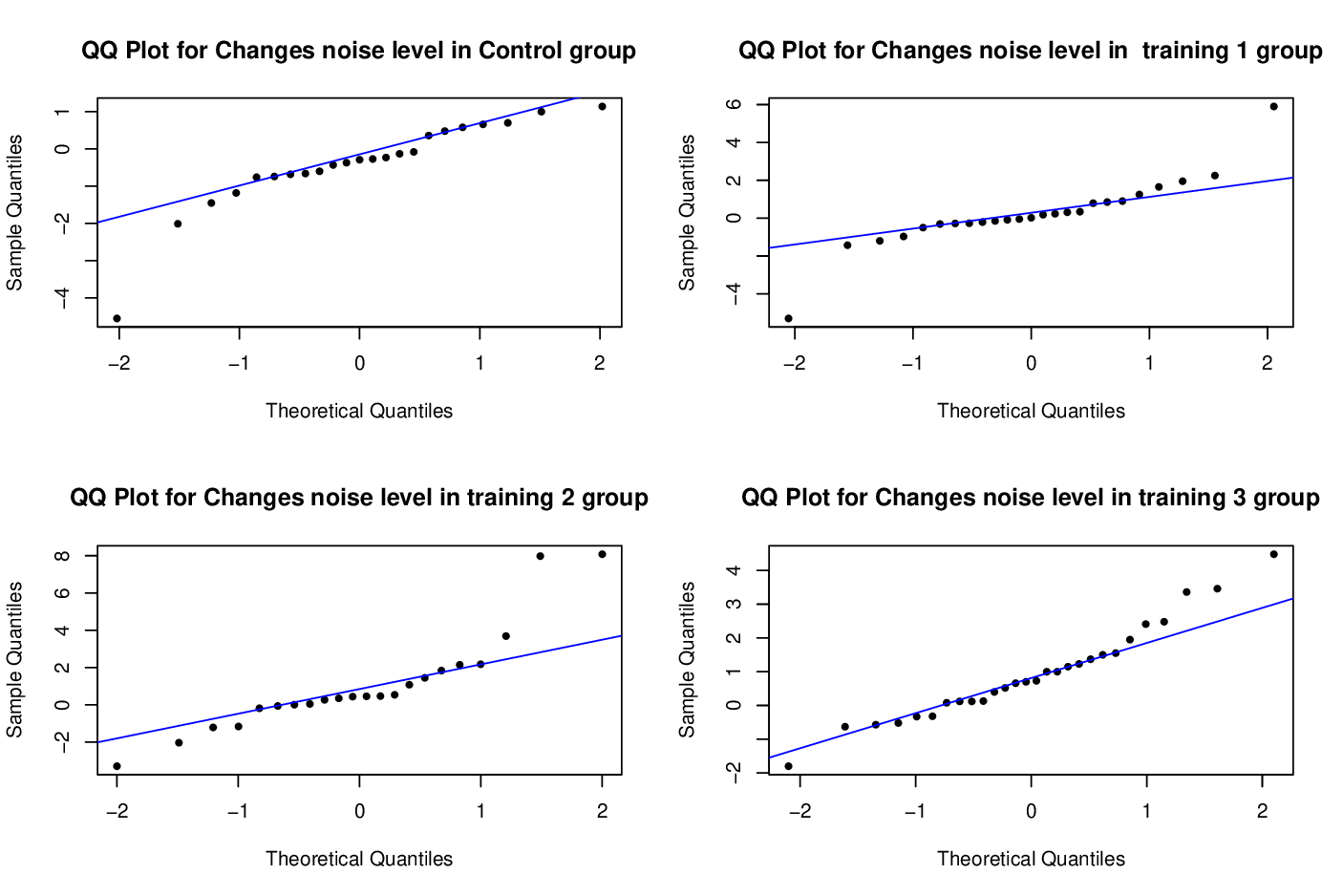}
    \caption{Q-Q plot of data}
    \label{fig:11}
\end{figure}

\section{Concluding Remarks}\label{Conclusions}

Three test procedures - Likelihood ratio test, Max-D test, and Min-D test are proposed to test the equality of mean effects against the tree ordered mean effects in a one-way analysis of variances under heteroscedasticity. To perform LRT, we have developed a numerical scheme to compute the MLEs of parameters under the tree ordered alternative. Two additional tests, referred to as Max-D and Min-D, are suggested, which are based on pairwise differences between the control mean and other treatment means. The parametric bootstrap is used to determine the critical values of tests. 
Extensive simulations are done to assess the size and power efficacy of these tests. It is observed that all three tests maintained a nominal significance threshold level close to $0.05$, with some exceptions noted for Min-D for small samples. In power comparison, we considered three different scenarios of tree ordered means; in each case, one of the tests outperformed the other two. One existing PB test for a related alternative hypothesis is also considered for comparisons and, our tests are shown to be more powerful. Additionally, we conducted the robustness study of these tests for different non-normal distributions (skew-normal, t, Laplace, mixture-normal, and exponential). All three tests preserve the nominal type I error rate for all symmetric and skew-normal distributions, with some exceptions noted for some small sample cases. The power plots show no loss of power for all tests for skew-normal, t, and Laplace distributions. A software package is developed in `R' and shared on the open platform `GitHub' for implementation on real data sets. Finally, the tests have been illustrated with an experimental study of patients undergoing some psychological treatment.

\backmatter

\section*{Declarations}

\textbf{Conflict of interest} The authors declare that they have no conflicts of interest relevant to the content of this article.

\bmhead{Supplementary information}
\section*{Acknowledgements}
The first author expresses gratitude for the financial support received from the ``Prime Minister Research Fellow (PMRF) Scheme, Government of India."

\section*{Data Availability Statement}  
The dataset analyzed in this study is publicly available at {\small \burl{https://www.kaggle.com/datasets/utkarshx27/headache-sufferers-for-sensitivity-to-noise?resource=download}}.

\begin{appendices}

\section{}\label{secA1}

We recall a few results regarding isotonic regression. Let ${\underline{\bf u}}, {\underline{\bf v}}$ be two vectors on $\mathbb{R}^{k+1}$ with the inner product and norm defined as 
\begin{equation}
    \left\langle \underline{\bf u},\underline{\bf v}\right\rangle_{\underline{\bf w}} = \sum\limits_{i=0}^{k} u_iv_iw_i \;\;\text{and}\;\; ||\underline{\bf u}||_{\underline{\bf w}} = \left\langle \underline{\bf u},\underline{\bf u}\right\rangle_{\underline{\bf w}}^{1/2},\label{A1}
\end{equation}
where $\underline{\bf u} = (u_0,u_1,\ldots u_k)^T$, $\underline{\bf v} = (v_0,v_1,\ldots v_k)^T$, and the weight function $\underline{\bf w} = (w_0,w_1,\ldots w_k)^T$ with $w_i>0$ for all $i=0,1,\ldots k.$ A vector $\underline{\bf u}^*$ is an isotonic regression of $(\underline{\bf u},\underline{\bf w})$ if and only if $\underline{\bf u}^*$ is the projection of $\underline{\bf u}$ onto $C$ under the inner product (\ref{A1}), where $C$ is the closed convex cone determined by a specified partial order. By Theorem 7.8 in  \cite{brunk1972statistical},  $\underline{\bf u}^*$ is an isotonic regression of $(\underline{\bf u},\underline{\bf w})$ if and only if $\underline{\bf u}^* \in C$ and 
\begin{equation}
    \left\langle \underline{\bf u}-\underline{\bf u}^*,\underline{\bf u}^*\right\rangle_{\underline{\bf w}}=0; \;\;\; \left\langle \underline{\bf u}-\underline{\bf u}^*,\underline{\bf z}\right\rangle_{\underline{\bf w}}\leq 0 \;\;\; \text{for any} \;\; \underline{\bf z}\in C. \label{A2}
\end{equation}

\cite{shi1998maximum} have given an algorithm for finding the MLEs of means and variances when means are under simple order restrictions. They have proved the convergence of their algorithm. However, in our case, we have tree order alternatives among means. The convergence of Algorithm \ref{Algorithm 1} can be established following some different steps. \\

{\bf Proof of the Theorem \ref{Theorem 1}:}
\begin{proof}
    From Algorithm \ref{Algorithm 1} and \cite{shi1998maximum}, the sequence $\{\left(\underline{\boldsymbol{\mu}}^{(m)}, \underline{\boldsymbol{\sigma}}^{2(m)}\right)\}$ is uniformly bounded. Therefore there exists a subsequence $\{\left(\underline{\boldsymbol{\mu}}^{(m_q)}, \underline{\boldsymbol{\sigma}}^{2(m_q)}\right)\}$ and a point $\left(\underline{\boldsymbol{\mu}}^*, \underline{\boldsymbol{\sigma}}^{2}(\underline{\boldsymbol{\mu}}^*)\right)\in I_0\times J_0$, where $I_{0}=\left\{\underline{\boldsymbol{\mu}} \in I : a \leq \mu_{i} \leq b, i=0,1, \ldots, k\right\}$ and $J_0$ is a compact subset in $\mathbb{R}^{k+1}$ such that \begin{align}
        \left(\underline{\boldsymbol{\mu}}^{(m_q)}, \underline{\boldsymbol{\sigma}}^{(m_q)}\right) \rightarrow \left(\underline{\boldsymbol{\mu}}^*, \underline{\boldsymbol{\sigma}}^{2}(\underline{\boldsymbol{\mu}}^*)\right)\; \text{as} \;m_q \rightarrow \infty .\label{A3}
    \end{align}
By Lemma A.1. of \cite{shi1998maximum},
\begin{align}
\underline{\boldsymbol{\mu}}^{(m_q-1)}-\underline{\boldsymbol{\mu}}^{(m_q)}\rightarrow \underline{\boldsymbol{0}}\;\;\text{and}\;\;\underline{\boldsymbol{\sigma}}^{2(m_q-1)}-\underline{\boldsymbol{\sigma}}^{2(m_q)}\rightarrow \underline{\boldsymbol{0}}. \label{A4}
\end{align}
Now we show that $\left(\underline{\boldsymbol{\mu}}^*, \underline{\boldsymbol{\sigma}}^{2}(\underline{\boldsymbol{\mu}}^*)\right)$ is the MLE of $\left(\underline{\boldsymbol{\mu}}, \underline{\boldsymbol{\sigma}}^2\right).$ Following \cite{brunk1972statistical} we have to show \begin{align}
\left\langle \underline{\overline{\bf{x}}}-\underline{\boldsymbol{\mu}}^*,{\underline{\boldsymbol{\mu}}}^*\right\rangle_{\underline{\bf{w}}^*}=0 \;\;\text{and}\;\; \left\langle \underline{\overline{\bf{x}}}-\underline{\boldsymbol{\mu}}^*,\underline{\boldsymbol{\mu}}\right\rangle_{\underline{\bf{w}}^*}\leq 0 \;\;\; \text{for any} \;\; \underline{\boldsymbol{\mu}} \in I_0. \label{A5}
\end{align}

Then
\begin{align*}
    \left\langle \underline{\overline{\bf{x}}}-\underline{\boldsymbol{\mu}}^*,\underline{\boldsymbol{\mu}}^*\right\rangle_{\underline{\bf{w}}^*} &= \lim_{m_q \to \infty}  \left\langle \underline{\overline{\bf{x}}}-\underline{\boldsymbol{\mu}}^{(m_q)},\underline{\boldsymbol{\mu}}^{(m_q)}\right\rangle_{\underline{\bf{w}}^{(m_q)}} \\
    &= \lim_{m_q \to \infty}  \left\langle \underline{\overline{\bf{x}}}-\underline{\boldsymbol{\mu}}^{(m_q)},\underline{\boldsymbol{\mu}}^{(m_q)}\right\rangle_{\underline{\bf{w}}^{(m_q-1)}}\;\;(\text{using (\ref{A4})}) \\
    &= 0 \;\;(\text{using\; (\ref{A2})})
\end{align*}
Now for any $\underline{\boldsymbol{\mu}} \in I_0$
\begin{align*}
    \left\langle \underline{\overline{\bf{x}}}-\underline{\boldsymbol{\mu}}^*,\underline{\boldsymbol{\mu}}\right\rangle_{\underline{\bf{w}}^*} &= \lim_{m_q \to \infty}  \left\langle \underline{\overline{\bf{x}}}-\underline{\boldsymbol{\mu}}^{(m_q)},\underline{\boldsymbol{\mu}} \right\rangle_{\underline{\bf{w}}^{(m_q)}} \\
     &= \lim_{m_q \to \infty}  \left\langle \underline{\overline{\bf{x}}}-\underline{\boldsymbol{\mu}}^{(m_q)},\underline{\boldsymbol{\mu}}\right\rangle_{\underline{\bf{w}}^{(m_q-1)}}\;\;(\text{using (\ref{A4})}) \\
    &\leq  0 \;\;(\text{using (\ref{A2})})
\end{align*}
Hence $\left(\underline{\boldsymbol{\mu}}^*, \underline{\boldsymbol{\sigma}}^{2}(\underline{\boldsymbol{\mu}}^*)\right)$ is the MLE of $\left(\underline{\boldsymbol{\mu}},\underline{\boldsymbol{\sigma}}^2\right).$ We proved that any convergence subsequence of $\{\left(\underline{\boldsymbol{\mu}}^{(m)}, \underline{\boldsymbol{\sigma}}^{2(m)}\right)\}$ converges to that MLE and hence $\{\left(\underline{\boldsymbol{\mu}}^{(m)}, \underline{\boldsymbol{\sigma}}^{2(m)}\right)\}$ converges to the MLE $\left(\underline{\boldsymbol{\mu}}^*, \underline{\boldsymbol{\sigma}}^{2}(\underline{\boldsymbol{\mu}}^*)\right)$. 
\end{proof}
\end{appendices}

\bibliographystyle{agsm}
\bibliography{ref_tree}

\end{document}